\documentclass[%
 aip,
 amsmath,amssymb,
 10pt, twocolumn,
floatfix
]{revtex4-1}

\usepackage{dcolumn}% Align table columns on decimal point
\usepackage{bm}% bold math

\pdfoutput=1
\usepackage{graphicx}

\usepackage{amssymb,amsmath,amsthm}
\usepackage{bm}

\newcommand{\pa}{\partial}
\newtheorem*{lem}{Lemma}

\begin{document}

\title{Spectral Analysis of the Koopman Operator for Partial Differential Equations}

\author{Hiroya Nakao}
\email{nakao@sc.e.titech.ac.jp}
\affiliation{Department of Systems and Control Engineering, Tokyo Institute of Technology, Tokyo 152-8552, Japan}

\author{Igor Mezi{\'c}}
\email{mezic@ucsb.edu}
\affiliation{Department of Mechanical Engineering and Mathematics, University of California, Santa Barbara, CA 93106, USA}

\date{November 13, 2020}% It is always \today, today,
             %  but any date may be explicitly specified

\begin{abstract}

We provide an overview of the Koopman operator analysis for a class of partial differential equations describing relaxation of the field variable to a stable stationary state. We introduce Koopman eigenfunctionals of the system and use the notion of conjugacy to develop spectral expansion of the Koopman operator. For linear systems such as the diffusion equation, the Koopman eigenfunctionals can be expressed as linear functionals of the field variable. The notion of inertial manifolds is shown to correspond to joint zero level sets of Koopman eigenfunctionals, and the notion of isostables is defined as the level sets of the slowest decaying Koopman eigenfunctional. Linear diffusion equation, nonlinear Burgers equation, and nonlinear phase-diffusion equation are analyzed as examples. 

\end{abstract}

\maketitle

\begin{quotation}
The Koopman operator approach to nonlinear dynamical systems has received considerable attention recently, which focuses on the evolution of the observables rather than on the system state itself and provides a rigorous method for globally linearizing the system dynamics. However, many studies of the Koopman operator analysis have focused only on finite-dimensional dynamical systems described by ordinary differential equations or maps.  We here show that the Koopman operator approach can be formally generalized to infinite-dimensional dynamical systems described by partial differential equations, providing new perspectives on the analysis and control of their nonlinear spatiotemporal dynamics. 
\end{quotation}

\section{Introduction}

Recent developments in the operator-theoretic approach to dynamical systems have provided new perspectives on the analysis of their properties~\cite{MezicBook,mauroy2020introduction, budivsic2012applied,mezic2013analysis,lan2013linearization,mezic2017koopman,korda2017data,surana2016linear,gaspard2005chaos}.
By using eigenfunctions of the Koopman operator originally introduced in the 1930's~\cite{koopman1931hamiltonian,vonneumann1932}, which describes evolution of observables rather than the system state itself, the system dynamics can often be decomposed into linearly independent Koopman modes even if the system is nonlinear.
In particular, if the dynamics is ergodic but non-chaotic, the spectrum of the Koopman operator in properly defined spaces does not contain (absolutely or singularly) continuous parts~\cite{mezic2017koopman,korda2017data} and the observable of the system can be represented as a linear combination of eigenfunctions associated with discrete eigenvalues of the Koopman operator.
In this sense, the Koopman-operator approach gives a general method for rigorously linearizing nonlinear dynamical systems~\cite{lan2013linearization,mezic2017koopman}, which can be used to develop new methods for the analysis and control of their behavior~\cite{mauroy2020introduction,surana2016linear}.

When the system exhibits stable limit-cycle oscillations, the Koopman operator of the system, restricted to the space of $L^2$ functions on a circle, has a ``basic'' pair of pure imaginary eigenvalues (as well as integer multiples of these) and the level set of the associated eigenfunction gives the {\em isochron} (equal-phase set) of the limit cycle, a key concept in the classical phase reduction theory for limit-cycle oscillators~\cite{winfree2001geometry,kuramoto2012chemical,hoppensteadt2012weakly,ermentrout2010mathematical,nakao2016phase} that gives foliation of the system dynamics around the limit-cycle solution.
In a similar spirit, it has recently been proposed that the Koopman eigenfunction associated with the largest eigenvalue can be used to define the {\em isostable} for a dynamical system converging to a stable fixed point, which is the set of system states that share equal-timing approach to the stable fixed point~\cite{mauroy2013isostables, mauroy2014converging, wilson2015extending, wilson2016isostable}.

In formulating the Koopman operator analysis, the dynamical systems studied so far have mainly been restricted to finite-dimensional systems described by ordinary differential equations (ODE's) or maps,
though the framework was originally set in the general context of nonlinear evolution equations in the Hilbert space including partial differential equations (PDE's)~(see Sec. 3.1 of Ref.~\cite{mezic2005spectral} where Koopman modes are defined). However, as recently discussed in several papers, the Koopman operator analysis and the notion of isostables can also be generalized to infinite-dimensional dynamical systems described by PDE's~\cite{wilson2016isostable,page2018koopman,kutz2018koopman}. Also, though not in the context of Koopman operator analysis, extension of the notion of the isochron to PDE's has been discussed in the phase reduction analysis of rhythmic reaction-diffusion systems~\cite{nakao2014phase}.

For a system described by a PDE, the system state is a field variable and thus the state space can be infinite-dimensional. Accordingly, the notion of Koopman eigenfunction of the system state for finite-dimensional systems should be generalized to {\em eigenfunctional} of the field variable.
The aim of this paper is to provide an overview of the Koopman operator analysis for a general class of PDE's describing relaxation of the system state to a stable stationary state (fixed point) in the simple case.
The key to our analysis here is the definition of Koopman eigenfunctionals.
It is shown that, by introducing Koopman eigenfunctionals of the system states, the Koopman operator analysis can naturally be generalized to such PDE's
on the basis of nonlinear functional analysis~\cite{Teschl}.

We use the idea of conjugacy, which has been used in Koopman mode decomposition earlier in the context of ordinary differential equations~\cite{lan2013linearization,mezic2017koopman,bollt2018matching} and data fusion \cite{williams2015data}, to develop spectral expansion of the Koopman operator. 
We also discuss the relationship of spectral expansions of Koopman operators for PDE's to the concept of inertial manifolds~\cite{foias1988inertial,foias1988computation,robinson2001infinite,jolly2001accurate} and argue that under certain conditions the inertial manifolds can be obtained as zero level sets of Koopman eigenfunctionals, enabling their computation.
These ideas are exemplified using conjugacy between linear diffusion equation and nonlinear Burgers and phase-diffusion equations.

Although the dynamic mode decomposition (DMD)~\cite{page2018koopman,mauroy2020introduction,schmid2010dynamic,kutz2018koopman,rowley2009spectral}, which is closely related to the Koopman-operator analysis, has widely been used for analyzing spatiotemporal dynamics of PDE's in recent years, explicit formulations of the Koopman-operator theory for PDE's have yet to be developed. Our aim in this paper is to provide an overview of the Koopman-operator approach to PDE's in the simple case, which would serve as a starting point for more rigorous mathematical analysis as well as to various applications of the Koopman-operator analysis for PDE's that arise in various areas of science and engineering.

This paper is organized as follows. In Sec. II, the Koopman formalism for PDE's whose solution converges to a uniform stationary state in introduced. In Sec. III, Koopman eigenfunctionals of linear PDE's are derived. In Sec. IV, three examples of solvable linear and nonlinear PDE's are analyzed in the context of conjugacy and compared with numerical simulations. Section V gives a summary and Appendix provides details of the calculations.

\section{Koopman operator formalism}

\subsection{System dynamics}

For simplicity, we consider a spatially one-dimensional, autonomous scalar PDE,
\begin{align}
\frac{\pa}{\pa t} u(x, t) = {\cal F}\{ u(x, t) \},
\label{PDE}
\end{align}
where $u \in C$ is the system state, i.e., the field variable representing a spatial pattern of the system, $t \in {\mathbb R}$ is the time, and $x \in [0, L]$ is the spatial coordinate.
The space $C$ is an appropriate function space on $[0, L]$, such as $L^2([0,L])$ of square-integrable functions.
The right-hand side ${\cal F}$ describes the dynamics of the system, which includes functions of $u(x, t)$ and its partial derivatives.
We assume appropriate boundary conditions at $x=0$ and $x=L$, such as Dirichlet, Neumann, or cyclic.
We also assume that Eq.~(\ref{PDE}) has an exponentially stable and isolated
stationary uniform solution $u_0(x) = 0$ ($0\leq x \leq L$) satisfying ${\cal F}\{ u_0(x) \} = 0$.
If $u_0(x)$ is non-uniform, we redefine $u(x, t) - u_0(x)$ as new $u(x, t)$.
We denote the basin of attraction of this stationary solution as ${\cal B} \subset C$.

In this paper, we simply assume that Eq.~(\ref{PDE}) with given boundary conditions is well-posed and has a unique solution for a given initial condition. We focus on the basin of attraction of an isolated stationary solution. If Eq.~(\ref{PDE}) has multiple stationary solutions, the analysis should be performed separately for each basin of attraction. The validity of these assumptions have to be individually analyzed for given PDE's.

In Sec.~III, we consider the diffusion equation,
\begin{align}
{\cal F}\{ u(x, t) \} = \frac{\pa^2 u(x, t)}{\pa x^2},
\end{align}
the Burgers equation,
\begin{align}
{\cal F}\{ u(x, t) \} = - u(x, t) \frac{\pa u(x, t)}{\pa x} + \frac{\pa^2 u(x, t)}{\pa x^2},
\end{align}
and the nonlinear phase-diffusion equation,
\begin{align}
{\cal F}\{ u(x, t) \} = \frac{\pa^2 u(x, t)}{\pa x^2} + \left( \frac{\pa u(x, t)}{\pa x} \right)^2,
\end{align}
as examples, which are mutually conjugate in the sense explained later.

\subsection{Koopman operator and infinitesimal generator}

In the Koopman operator analysis, evolution of the observables of the system rather than the system state itself is focused on~\cite{MezicBook,budivsic2012applied,mezic2013analysis}. Even if the system dynamics is nonlinear, the evolution of the observables is described by a linear Koopman operator.
For ODE's, the observable is generally a nonlinear function that maps the system states represented by finite-dimensional vectors to complex values. For PDE's, the observable is generally a nonlinear functional that maps the system states represented by functions, i.e., field variables, to complex values. For rigorous mathematical foundations of the nonlinear functional analysis, see Ref.~\cite{Teschl}.

For systems described by PDE's, the observable of the system is given by an observation functional
\begin{align}
g[u] : {\cal B} \to {\mathbb C}
\end{align}
of the field variable $u \in {\cal B}$. Evolution of the observable is described by a Koopman operator $U^t$, satisfying
\begin{align}
U^t g[u] = g[ S^t u ],
\end{align}
where $U^t g$ is the observable at time $t$ starting from $g$ at time $0$ and $S^t u$ is the field variable at  $t$ starting from $u$ at time $0$, respectively. Here, $S^t : {\cal B} \to {\cal B}$ is the flow of the PDE~(\ref{PDE}), which satisfies
\begin{align}
S^t u(x, s) = u(x, s+t)
\end{align}
for arbitrary $s$ and $t$. It is clear that $U^0$ is an identity, $U^0 g[u] = g[S^0 u] = g[u]$, $U^t$ is a linear operator, 
\begin{align}
U^t ( c_1 g_1[u] + c_2 g_2[u] )
&= c_1 g_1[S^t u] + c_2 g_2[S^t u]
\cr
&= c_1 U^t g_1[u] + c_2 U^t g_2[u]
\end{align}
for any $c_1$ and $c_2$, and $U^t$ is commutable, as $g[S^{s} S^tu]=U^t U^s g[u]$ and $g[S^{s} S^tu]=g[S^{t} S^su]$
imply
\begin{align}
U^t U^s g[u]=U^s U^t g[u].
\end{align}
The inverse operator $U^{-t} = (U^{t})^{-1}$ exists as long as the flow of the PDE~(\ref{PDE}) is a group over $t$.

Infinitesimal evolution of a smooth observable $U^t g$ at $t$ is represented as
\begin{align}
\frac{d}{dt} \{ U^t g[u] \} = A \{ U^t g[u] \},
\end{align}
where the linear operator $A$ is an infinitesimal generator of the Koopman operator $U^{t}$ given by
\begin{align}
A g[u] 
= \lim_{\tau \to 0} \frac{ U^{\tau} g[u] - g[u] } {\tau}
= \int_0^L {\cal F}\{ u(x) \} \frac{\delta g[u]}{\delta u(x)} dx.
\label{generator-def}
\end{align}
Here, $\delta g[u] / \delta u(x)$ is a functional derivative of $g[u]$ with respect to $u(x)$ (see Appendix for details). Using the generator $A$, the action of the Koopman operator $U^t$ on the observable $g[u]$ can be expressed as
\begin{align}
U^t g[u] = \exp( A t ) g[u] = \sum_{k=0}^{\infty} \frac{1}{k!} t^k A^k g[u].
\end{align}

\subsection{Koopman eigenfunctionals}

Since $U^t$ is a linear operator, we can consider its eigenvalue $\lambda \in {\mathbb C}$ and eigenfunctional $\phi_\lambda[u] : {\mathcal B} \to {\mathbb C}$ satisfying
\begin{align}
U^t \phi_{\lambda}[u] 
= \phi_{\lambda}[ S^t u ]
= e^{\lambda t} \phi_{\lambda}[ u ],
\end{align}
where $\lambda$ rather than $e^{\lambda t}$ is called the eigenvalue
(this term is reserved for the generator $A$ of $U^t$ introduced below).
If $\phi_{\lambda_1}[u]$ is an eigenfunctional associated with $\lambda_1$
and $\phi_{\lambda_2}[u]$ is an eigenfunctional associated with $\lambda_2$,
the product $\phi_{\lambda_1}[u] \phi_{\lambda_2}[u]$ is also an eigenfunctional associated with the eigenvalue $\lambda_1 + \lambda_2$, because
\begin{align}
U^t ( \phi_{\lambda_1}[u] \phi_{\lambda_2}[u] )
&= \phi_{\lambda_1}[S^t u] \phi_{\lambda_2}[S^t u]
\cr
&= e^{\lambda_1 t} \phi_{\lambda_1}[u]
e^{\lambda_2 t} \phi_{\lambda_2}[u]
\cr
&= e^{(\lambda_1 + \lambda_2)t} \phi_{\lambda_1}[u] \phi_{\lambda_2}[u].
\end{align}
Similarly, $\phi_\lambda^n[u] = ( \phi_{\lambda}[u] )^n$ is an eigenfunctional of $U^t$ with eigenvalue $n \lambda$,
and more generally
\begin{align}
\phi_{\lambda_1}^{k_1}[u] \phi_{\lambda_2}^{k_2}[u] \cdots \phi_{\lambda_n}^{k_n}[u]
\label{products}
\end{align}
with non-negative integers $k_1, \cdots, k_n$ satisfying $k_1 + \cdots + k_n > 0$ is an eigenfunctional of $U^t$ associated with an eigenvalue $\sum_{j=1}^n k_j \lambda_j$ for $n=2, 3, \cdots$.

A smooth eigenfunctional $\phi_{\lambda}[u]$ of the Koopman operator $U^t$ is also an eigenfunctional of the generator $A$, i.e.,
\begin{align}
A \phi_{\lambda}[u] = \lambda \phi_{\lambda}[u],
\label{eigen_generator}
\end{align}
because
\begin{align}
A \phi_{\lambda}[u]
&= \lim_{\tau \to 0} \frac{ U^\tau \phi_{\lambda}[u] - \phi_{\lambda}[u] }{\tau }
\cr
&= \lim_{\tau \to 0} \frac{ e^{\lambda \tau} - 1 }{\tau } \phi_{\lambda}[u] = \lambda \phi_{\lambda}[u].
\label{generatoreigen}
\end{align}

For the stationary state $u_0(x)$ satisfying ${\cal F}\{ u_0(x) \} = 0$, 
$S^t u_0(x) = u_0(x)$ holds for any $t \geq 0$ and thus
\begin{align}
U^t \phi_{\lambda}[u_0] &= \phi_\lambda [ S^t u_0 ] = \phi_\lambda [ u_0 ]
\end{align}
and
\begin{align}
A \phi_{\lambda}[u_0] &= 0.
\end{align}
Therefore, $\phi_{\lambda}[u_0]$ can take non-zero values only when $\lambda = 0$. As long as $\lambda \neq 0$ ($e^{\lambda t} \neq 1$), we have $\phi_\lambda[u_0] = 0$.

If the system has a conserved quantity represented by a functional $h[u]$, then $h[u]$ satisfies
\begin{align}
U^t h[u] = h[S^t u] = h[u]
\end{align}
and
\begin{align}
A h[u] = 0.
\end{align}
Thus, $h[u]$ is a Koopman eigenfunctional of $U$ and $A$ associated with a vanishing eigenvalue $\lambda = 0$ ($e^{\lambda t} = 1$).

Note that the system can generally possess two or more conserved quantities.
For example, for the 2D Euler equation with appropriate boundary conditions, the energy (2nd moment of the velocity field) and enstrophy (2nd moment of the vorticity field) are both conserved and hence they are Koopman eigenfunctionals with eigenvalue $0$.

\subsection{Linear systems}

For linear PDE's, we can derive an explicit formula for the Koopman eigenvalues and eigenfunctionals. Consider a linear PDE,
\begin{align}
\frac{\pa}{\pa t} u(x, t) = {\cal L} u(x, t),
\label{linearpde}
\end{align}
where ${\cal L}$ is a linear time-independent operator describing the evolution of the system.
We introduce appropriate boundary conditions and assume that Eq.~(\ref{linearpde}) has a stable, isolated, stationary uniform solution $u_0(x) = 0$ for $0 \leq x \leq L$.
We assume that ${\cal L}$ has a discrete spectrum with eigenvalues $\lambda = \lambda_n$ ($n=1, 2, \cdots, N$) with negative real part, where $N$ can go to infinity.
The case with zero eigenvalue is excluded because we consider an isolated stable stationary state. We denote the flow of Eq.~(\ref{linearpde}) as 
\begin{align}
S_{\cal L}^t = e^{{\cal L}t}, 
\end{align}
and the corresponding Koopman operator as
\begin{align}
U^t_{\cal L} = \exp (A_{\cal L} t),
\end{align}
where $A_{\cal L}$ is the generator of $U^t_{\cal L}$.

We denote the eigenfunction of ${\cal L}$ associated with eigenvalue $\lambda$ as $q_{\lambda}$, i.e., 
\begin{align}
{\cal L} q_{\lambda}(x) = \lambda q_{\lambda}(x),
\end{align}
and the eigenfunctions of the adjoint operator ${\cal L}^*$ of ${\cal L}$ associated with eigenvalue $\overline{\lambda}$ as $w_{\lambda}$ (the overline indicates complex conjugate),
\begin{align}
{\cal L}^{*} w_{\lambda}(x) = \overline{\lambda} w_{\lambda}(x),
\label{adjoint}
\end{align}
where the adjoint operator ${\cal L}^{*}$ of ${\cal L}$ is defined with respect to the inner product $\int_0^L a(x) \overline{b(x)} dx$ of two functions $a(x)$ and $b(x)$, 
and appropriate adjoint boundary conditions for $w_{\lambda}(x)$ should be introduced so that the bilinear concomitant vanishes~\cite{keener1988principles}.

Using $w_{\lambda}$ as a weight function, we introduce a linear functional 
\begin{align}
\phi_{\lambda}[u] = \int_0^L u(x) \overline{w_{\lambda}(x)} dx.
\label{linearform}
\end{align}
This $\phi_{\lambda}[u]$ gives the eigenfunctional of the generator $A_{\cal L}$ with eigenvalue $\lambda$, because
\begin{align}
A_{\cal L} \phi_{\lambda}[u]
&= \int \{ {\cal L} u(x) \} \frac{\delta \phi_{\lambda}[u]}{\delta u(x)} dx
\cr
&= \int \{ {\cal L} u(x) \} \overline{w_{\lambda}(x)} dx 
\cr
&= 
\int u(x) \{ \overline{ {\cal L}^{*} w_{\lambda}(x) } \} dx
\cr
&= \lambda \int u(x) \overline{w_{\lambda}(x)} dx
= \lambda \phi_{\lambda}[u]
\label{adjoint0}
\end{align}
from Eq.~(\ref{generator-def}) and Eq.~(\ref{eigen_generator}).
The above $\phi_{\lambda}[u]$ is also an eigenfunctional of the Koopman operator $U^t_{\cal L}$, because
\begin{align}
U_{\cal L}^t \phi_{\lambda}[u] 
&= \phi_{\lambda_n}[S_{\cal L}^t u] = \phi_{\lambda_n}[e^{{\cal L}t} u]
\cr
&=\int \{ e^{{\cal L}t} u(x) \} \overline{w_{\lambda}(x)}dx\nonumber \\
&=\int u(x) \overline{e^{{\cal L^*}t} w_{\lambda}(x)}dx\nonumber \\
&= e^{\lambda t} \int u(x) \overline{w_{\lambda}(x)}dx\nonumber 
= e^{\lambda t} \phi_{\lambda}[u].
\end{align}
In Sec. III, we derive Koopman eigenfunctionals for a linear diffusion equation using the above result.

\subsection{Conjugacy}

Generally, the spectrum of the Koopman operator $U^t$ may contain discrete and continuous parts.
When the solution of the PDE converges to a stationary solution $u_0(x)$ globally and exponentially in some Banach space, the eigenvalues of $U^t$ have negative real part and, under certain conditions, they include the eigenvalues of the linearized operator of ${\cal F}$ around $u_0(x)$.
To show this, we introduce the notion of conjugacy~\cite{MezicBook, budivsic2012applied,mezic2017koopman}.

Let $S_{\cal F}^t$ and $U_{\cal F}^t$ be the flow and Koopman operator for a PDE
\begin{align}
\frac{\pa}{\pa t} u(x, t) = {\cal F}\{ u(x, t) \},
\end{align}
and $S_{\cal G}^t$ and $U_{\cal G}^t$ the flow and Koopman operator for another PDE
\begin{align}
\frac{\pa}{\pa t} v(x, t) = {\cal G}\{ v(x, t) \},
\end{align}
where $u \in C$ and $v \in C$ are the field variables and ${\cal F}$ and ${\cal G}$ describe their dynamics.

These two systems are {\em conjugate} when a diffeomorphism, i.e., smooth mapping
\begin{align}
v = \Phi( u )
\end{align}
such that
\begin{align}
\Phi( S_{\cal F}^t u ) = S_{\cal G}^t \Phi( u )
\end{align}
exists. Suppose that $\phi_\lambda[v]$ is an eigenfunctional of $U_{\cal G}^t$ with eigenvalue $\lambda$. Then,
\begin{align}
U^t_{\cal G} \phi_\lambda[v] &= \phi_\lambda[ S_{\cal G}^t v ] = \phi_\lambda[ S_{\cal G}^t \Phi( u ) ]
\cr
&
= \phi_\lambda[ \Phi( S_{\cal F}^t u ) ] = U_{\cal F}^t \phi_\lambda[ \Phi( u ) ]
\cr
&= e^{\lambda t} \phi_\lambda[v] = e^{\lambda t} \phi_\lambda[ \Phi( u ) ],
\label{conjugate-eigen}
\end{align}
so $\phi_\lambda[\Phi(u) ]$ is an eigenfunctional of $U_{\cal F}^t$ of the conjugate system with the same eigenvalue $\lambda$.

Now, assume that a fully nonlinear problem
\begin{align}
\frac{\partial}{\partial t} u(x, t) = {\cal F}\{ u(x, t) \}
\label{fullynonlinear}
\end{align}
has a stable stationary solution $u_0(x) = 0$ and split ${\cal F}$ around $u_0(x)$ as
\begin{align}
{\cal F}\{ u(x,t) \} = {\cal L} u(x, t) + {\cal N}\{ u(x, t) \},
\end{align}
where ${\cal L}$ is a linearized operator of ${\cal F}$ around $u_0(x) = 0$ and ${\cal N}$ is a nonlinear part satisfying $ {\cal N}\{ u_0(x) \} = 0 $. We consider a linearized problem,
\begin{align}
\frac{\partial}{\partial t} v(x, t) = {\cal L} v(x, t) = {\cal G}\{ v(x, t) \},
\label{linearized}
\end{align}
and assume that ${\cal L}$ has a discrete spectrum with eigenvalues $\lambda_n (n=1, 2, \cdots)$ with negative real part.

As before, let $w_{\lambda}$ be the eigenfunctions of the adjoint operator ${\cal L^*}$ of ${\cal L}$, satisfying
\begin{align}
{\cal L}^* w_{\lambda}(x) = \overline{\lambda} w_{\lambda}(x).
\end{align}
Then, the functional
\begin{align}
\phi_{\lambda}[v]=\int v(x) \overline{w_{\lambda}(x)}dx
\end{align}
is an eigenfunctional of the Koopman operator $U_{\cal L}^t$ for the linear system Eq.~(\ref{linearized}) and satisfies
\begin{align}
U_{\cal L}^t \phi_{\lambda}[v]
= \phi_{\lambda}[e^{{\cal L}t} v]
&= e^{\lambda t} \phi_{\lambda}[v].
\end{align}

Now suppose that the solutions of the original fully nonlinear problem Eq.~(\ref{fullynonlinear})
and the linearized problem Eq.~(\ref{linearized}) are conjugate, i.e.,
\begin{equation}
\Phi ( S_{\cal F}^t u ) = e^{{\cal L} t} \Phi( u ),
\label{conj}
\end{equation}
where $S_{\cal F}^t$ is the flow of Eq.~(\ref{fullynonlinear}) and $S_{\mathcal G}^t = e^{{\cal L} t}$ is the flow of Eq.~(\ref{linearized}).
Then, from Eq.~(\ref{conjugate-eigen}), if $\phi_\lambda[v]$ is an eigenfunctional of $U_{\cal L}^t$ with eigenvalue $\lambda$, then $\phi_{\lambda}[\Phi( u )]$ is an eigenfunctional of $U_{\cal F}^t$ with the same eigenvalue $\lambda$. 
Thus, the eigenvalue $\lambda$ of the linearized operator ${\cal L}$ is also an eigenvalue of the Koopman operator $U_{\cal F}^t$.
These eigenvalues $\lambda_1, \lambda_2, ...$ of the linearized operator ${\cal L}$ and of the generator $A$ will be called {\em principal eigenvalues} in the following discussion. They are sorted in decreasing order of their real part as ${0 > }\mbox{Re} \lambda_1 \geq \mbox{Re} \lambda_2 \geq \cdots$ as usual.

\subsection{Expansion of observables}

We here discuss the expansion of analytic observables by the principal Koopman eigenfunctionals. We consider a system
\begin{align}
\frac{\pa}{\pa t} u(x, t) = {\cal F}\{ u(x, t) \}
\label{systemexpanded}
\end{align}
with an exponentially stable stationary state $u_0(x) = 0$, and denote the linearized operator of ${\cal F}$ around $u_0(x)$ as ${\cal L}$. We assume that the original and linearized systems are conjugate, and ${\cal L}$ has a discrete spectrum with eigenvalues $\lambda_j$ ($j=1, 2, \cdots$) that accumulate only at infinity, $\lim_{j \to \infty} \mbox{Re\ }\lambda_j = - \infty$.
These eigenvalues are then the principal eigenvalues of the Koopman operator $U^t_{\cal F}$.
We denote the eigenfunction of ${\cal L}$ associated with $\lambda_n$ as $q_n$ and the field variable $u$ can be expressed as
\begin{align}
u(x, t) = \sum_{j=1}^{\infty} a_j(t) q_j(x),
\end{align}
where $a_1, a_2, \cdots$ are the expansion coefficients.

Similarly to the case of ODE's~\cite{MezicBook}, when the solution of the PDE $u(x, t)$ converges to a globally and exponentially stable stationary solution $u_0(x) = 0$, a sufficiently smooth observable $g[u]$ can be expanded as an infinite series over the set of Koopman eigenfunctionals $\{ \phi_{\lambda_n} [u] \}$ with principal eigenvalues $\{ \lambda_n \}$ ($n=1, 2, ..., \infty$) as 
\begin{align}
g[u] = g[0] 
+ \sum_{n=1}^{\infty}
\frac{1}{n!}
&\sum_{j_1=1}^{\infty} \sum_{j_2=1}^{\infty} \cdots \sum_{j_n=1}^{\infty} c_{j_1, j_2, ..., j_n}
\cr
&\times
\phi_{\lambda_{j_1}}[u] \phi_{\lambda_{j_2}}[u] \cdots \phi_{\lambda_{j_n}}[u],
\end{align}
where $c_{j_1, j_2, \cdots, j_n}$ are expansion coefficients determined by functional derivatives of $g[u]$. See Appendix for the derivation and explicit expression for the coefficients.

Evolution of the observable $g[u]$ can be expressed as
\begin{align}
U^t g[u] = 
g[0] + \sum_{n=1}^{\infty}
\frac{1}{n!}
\sum_{j_1=1}^{\infty} \sum_{j_2=1}^{\infty} \cdots \sum_{j_n=1}^{\infty} c_{j_1, j_2, ..., j_n}
\cr
\times
e^{(\lambda_{j_1} + \lambda_{j_2} + \cdots + \lambda_{j_n}) t} \phi_{\lambda_{j_1}}[u] \phi_{\lambda_{j_2}}[u] \cdots \phi_{\lambda_{j_n}}[u],
\end{align}
where the summation is taken over all discrete eigenvalues. Thus, by regarding the Koopman eigenfunctionals as new coordinates, dynamics of the system can be decomposed into linear independent components.

In particular, if the observable is chosen as
\begin{align}
g_x[ u ] = u(x)
\end{align}
where the subscript $x$ indicates that the observable $g_x[ u ]$ gives the value of the function $u$ at $x$, i.e., the field variable $u(x)$ itself, the evolution of the system state can be expressed as
\begin{align}
&u(x, t) = g_x[ S^t u ] = U^t g_x [u]
\cr
&=
\sum_{n=1}^{\infty}
\frac{1}{n!}
\sum_{j_1=1}^{\infty} \sum_{j_2=1}^{\infty} \cdots \sum_{j_n=1}^{\infty} d_{j_1, j_2, ..., j_n}(x)
\cr
&\times e^{(\lambda_{j_1} + \lambda_{j_2} + \cdots + \lambda_{j_n}) t} 
\phi_{\lambda_{j_1}}[u] \phi_{\lambda_{j_2}}[u] \cdots \phi_{\lambda_{j_n}}[u],
\label{expansion}
\end{align}
where $g_x[0] = u_0(x) = 0$ is used and $d_{j_1, j_2, \cdots, j_n}(x)$ are $x$-dependent expansion coefficients of the observable $g_x[u]$, called the Koopman modes~\cite{mezic2017koopman}. The lowest order terms of this expansion is given by
\begin{align}
u(x, t) = \sum_{j_1=1}^{\infty} d_{j_1}(x) e^{\lambda_{j_1} t} \phi_{\lambda_{j_1}}[u]
+ \cdots.
\label{expansion2}
\end{align}

Here, we give a lemma, which is useful for considering the inertial manifolds and model reduction of the PDE.

\begin{lem}
Let the real parts $\sigma_j$ of the eigenvalues of ${\cal L}$  satisfy $\lim_{j \to \infty} \sigma_j= -\infty$, and let there be no finite accumulation points of the spectrum in the left half of the complex plane. Let $D<0$. Then, the number of eigenvalues of the Koopman operator $U^t$ such that $\sigma_j>D$ is finite.
\end{lem}

\begin{proof}
There is the smallest $K$ such that $\sigma_k<D$ for $k \geq K$. Also, for any $\sigma_j > \sigma_K$ there is a positive integer $n$ such that $n\sigma_j<D$. Since the eigenvalues of the Koopman eigenfunctional take the form $\sum_{j=1}^n k_j \lambda_j$ where $k_j \geq 0$ ($j=1 ..., n$), the number of eigenvalues whose real parts are bigger than $D$ is finite
\footnote{This result can easily be generalized for periodic attractors, where instead of
claiming the number of eigenvalues is finite in a subset of left half plane, we can claim
the number is finite in a rectangular subset of the complex plane that
includes the imaginary axis and $0$ on its right boundary.}.
\end{proof}

\subsection{Inertial manifold and isostables}

Using Eq.~(\ref{expansion}), we can introduce the notion of {\em inertial manifold}~\cite{foias1988inertial,foias1988computation,robinson2001infinite,jolly2001accurate} and {\em isostable}~\cite{mauroy2013isostables, mauroy2014converging, wilson2015extending, wilson2016isostable} of the system.
The existence of the inertial manifold has been discussed for several class of PDE's such as reaction-diffusion systems~\cite{foias1988inertial,foias1988computation,robinson2001infinite}.
Also, the notion of isostable has recently been applied to the control of ODE's and PDE's~\cite{mauroy2014converging,wilson2016isostable,wilson2020adaptive}, and also for phase-amplitude reduction of limit-cycling systems~\cite{wilson2015extending,mauroy2016global,shirasaka2017phase,shirasaka2020phase}. 

The inertial manifold of the system is a finite-dimensional smooth manifold ${\cal I}$, which satisfies $S^t {\cal I} {\subseteq} {\cal I}$ for $t>0$ and exponentially attracts all solutions of the system (hence ${\cal I}$ includes the globally stable stationary state).
When the Koopman eigenfunctionals of the system are known, the inertial manifold can be defined as joint zero level sets of the Koopman eigenfunctionals other than the one associated with the largest eigenvalue~\cite{MezicBook,mezic2017koopman}.
The isostable of a stable stationary state is the set of system states that share the same asymptotic convergence to the stationary state, i.e., that converge to the stationary state with the same timing as $t \to \infty$.
The isostable can be defined as a level set of the Koopman eigenfunctional associated with the eigenvalue with the largest real part~\cite{MezicBook,mauroy2014converging,wilson2016isostable}.

For example, if $\lambda_1$ is real and $0 > \lambda_1 > {\mbox Re} \lambda_2 \geq \cdots$, the lowest-dimensional inertial manifold can be defined as
\begin{align}
{\cal I} = \{ u(x) \in {\cal B} \ | \ \phi_{\lambda_n}[u] = 0, \ \forall n \geq 2 \},
\end{align}
and the evolution of the system state $u(x, t) = S^t u(x,0)$ on ${\cal I}$ starting from $u \in {\cal I}$ at time $0$ can be represented as
\begin{align}
u(x, t) = d_{1}(x) e^{\lambda_1 t} \phi_{\lambda_1}[u] + d_{1,1}(x) e^{2 \lambda_1 t} \{ \phi_{\lambda_1}[u] \}^2 + \cdots.
\label{inertiareal}
\end{align}
From Eq.~(\ref{expansion}), all system states in ${\cal B}$ are exponentially attracted to the one-dimensional inertial manifold ${\cal I}$ and converges to the stable stationary state along ${\cal I}$.

At sufficiently large $t$, the system state $u(x, t)$ starting from arbitrary initial state in ${\cal B}$ approximately satisfies Eq.~(\ref{inertiareal}).
Therefore, if two system states $u_1(x)$ and $u_2(x)$ satisfy $| \phi_{\lambda_1}[u_1] | = | \phi_{\lambda_1}[u_2] |$ initially, they converge to the stationary state $u_0(x) = 0$ with the same timing, satisfying $\| u_1 \| \simeq \| u_2 \|$ asymptotically as $t \to \infty$.
Here, the norm $\| a \|$ of a function $a(x)$ is defined as $\| a \| = ( \int_0^L |a(x)|^2 dx )^{1/2}$.
Thus, the level set of $\phi_{\lambda_1}[u]$ gives the isostable of the stationary state.\\

Note that the 1st Koopman mode $d_1(x)$ in Eq.~(\ref{inertiareal}) is simply proportional to the eigenfunction $q_1(x)$ of $\mathcal{L}$. At sufficiently large $t$, the system state $u(x, t)$ is close to $u_0(x) = 0$ and approximately obeys a linearized dynamics $\pa u(x, t) /\pa t = {\cal L} u(x, t)$. Thus, the solution at sufficiently large $t$ is dominated by the slowest decaying mode $q_1(x)$ with the largest $\lambda_1$ as $u(x, t) \approx a_1 q_1(x) e^{\lambda_1 t}$ where $a_1$ is a constant, and the 1st Koopman mode $d_1(x)$ can be obtained from the field variable $u(x, t)$ by
\begin{align}
d_1(x) \phi_{1}[u] = \lim_{t \to \infty} \{ u(x, t) e^{-\lambda_1 t} \} = a_1 q_1(x),
\end{align}
i.e. $d_1(x) \propto q_1(x)$.

Similarly, if $\lambda_1$ and $\lambda_2$ are complex conjugate mutually, then we have $0 > {\mbox Re} \lambda_1 = {\mbox Re} \lambda_2 > {\mbox Re} \lambda_3 \geq \cdots$.
We thus have a two-dimensional inertial manifold in this case, which can be defined as
\begin{align}
{\cal I} = \{ u(x) \in {\cal B} \ | \ \phi_{\lambda_n}[u] = 0, \ \forall n \geq 3 \}
\end{align}
and the evolution of the system state on the two-dimensional inertial manifold ${\cal I}$ can be expressed as
\begin{align}
u(x, t) &=
 d_{1}(x) e^{\lambda_1 t} \phi_{\lambda_1}[u] + d_{2}(x) e^{\lambda_2 t} \phi_{\lambda_2}[u]
 \cr
 +&
  d_{1,1}(x) e^{2 \lambda_1 t} \{ \phi_{\lambda_1}[u] \}^2 
	+ 
 d_{2,2}(x) e^{2 \lambda_2 t} \{ \phi_{\lambda_2}[u] \}^2
	\cr
  +&
  	d_{1,2}(x) e^{(\lambda_1 + \lambda_2) t} \phi_{\lambda_1}[u] \phi_{\lambda_2}[u]
	+ \cdots.
\end{align}
The level sets of $|\phi_{\lambda_1}[u]|$ (or equivalently $|\phi_{\lambda_2}[u]|$) give the isostables.

The above construction of the 1st Koopman mode is essentially the same as those used in the analysis of isostables for ODE's~\cite{mauroy2013isostables} and PDE's~\cite{wilson2016isostable}.\\

\section{Examples}

We here provide Koopman eigenfunctionals for three simple examples of linear and nonlinear PDE's
that are mutually conjugate, and illustrate by direct numerical simulations that the
values of the Koopman eigenfunctionals of the evolving system state
actually exhibit exponential relaxation.

\begin{figure}[htbp]
\includegraphics[width=0.85\hsize]{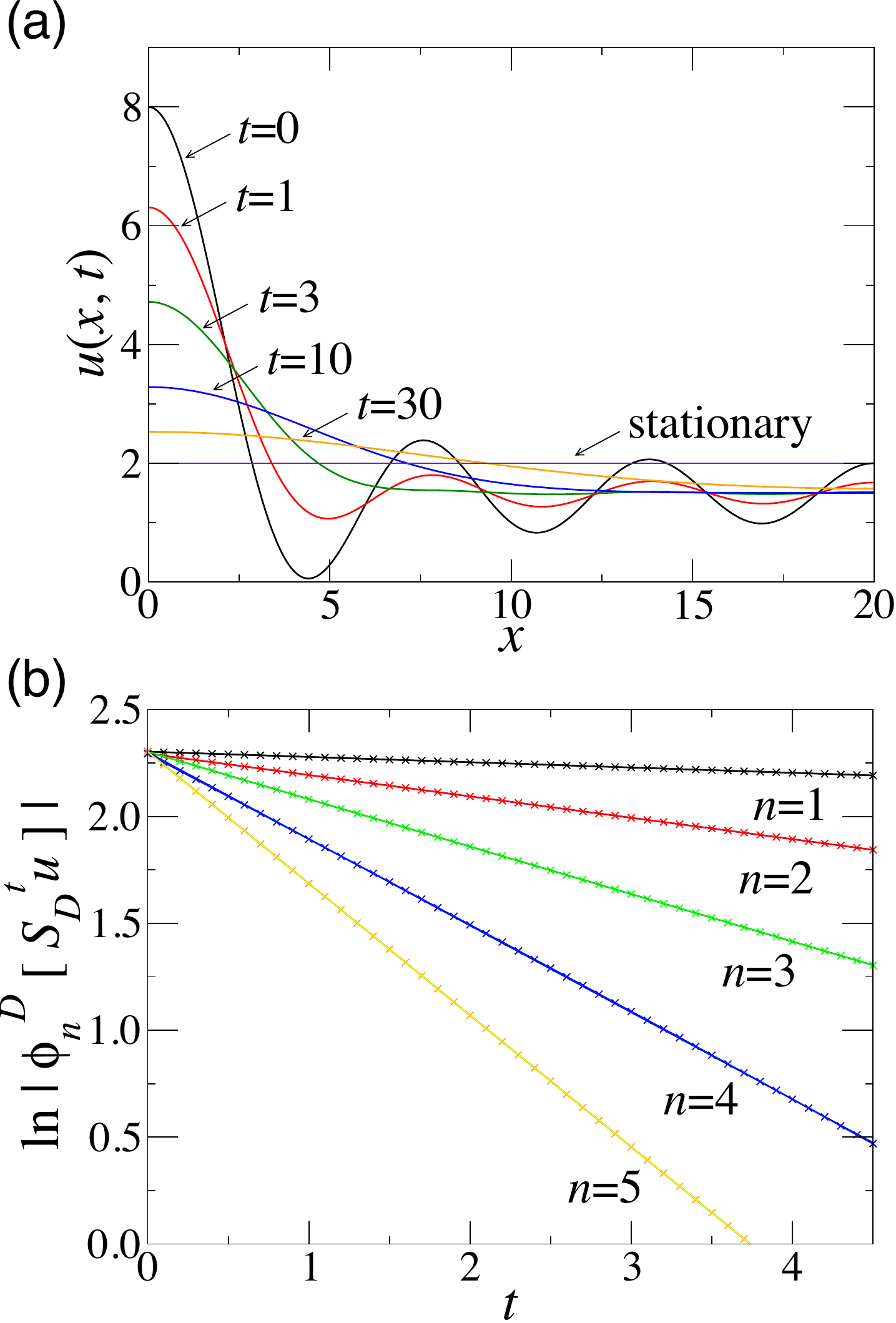}
\caption{Linear diffusion equation with Neumann boundary conditions. (a) Relaxation of the field variable $u(x, t)$ to a uniform stationary state. (b) Exponential decay of the Koopman eigenfunctionals $\phi^D_{n}[u(x,t)]$ for $n=1, \cdots, 5$. The initial state is $u(x, 0) =  2 + \cos( \pi  x / L ) + \cos( 2 \pi x / L ) + \cos( 3 \pi x / L ) + \cos( 4 \pi x / L ) + \cos( 5 \pi x / L ) + \cos( 6 \pi x/ L )$. }
\label{fig1}
\end{figure}

\subsection{Linear diffusion equation}

\subsubsection{Neumann boundary conditions}

We first consider a linear diffusion equation on $[0, L]$,
\begin{align}
\frac{\pa}{\pa t} u(x, t) = \frac{\pa^2}{\pa x^2} u(x, t),
\label{lineardiffusion}
\end{align}
where $u(x, t) \in L^2([0,L])$ is the field variable, and assume homogeneous Neumann boundary conditions,
\begin{align}
\frac{\partial u}{\partial x}(0, t) = \frac{\partial u}{\partial x}(L, t) = 0.
\label{neumann}
\end{align}
Then Eq.~(\ref{lineardiffusion}) has a stable, spatially-uniform stationary solution
\begin{align}
u_0(x) = \frac{1}{L} \int_0^L u(x{'}, 0) dx' =const. \quad (0 \leq x \leq L).
\end{align}
In the following discussion, we redefine $u(x, t) - u_0(x)$ as the new field variable $u(x, t)$. Then, the new $u(x, t)$ obeys the same diffusion equation~(\ref{lineardiffusion}) and the boundary conditions~(\ref{neumann}) and $u(x, t) \to 0$ as $t \to \infty$.
This eliminates the zero eigenvalue from the problem, which arises when Neumann boundary conditions are assumed and the integral of $u(x, 0)$ does not vanish.
Note that, if we do not subtract $u_0(x)$ from $u(x, t)$ in this problem, we have a continuous family of the final stationary solutions depending on the value of $\int_0^L u(x, 0) dx$, which is conserved through the evolution.
This is not the case if we add generic perturbations to the system.

We denote the flow, Koopman operator, and generator of the Koopman operator of this system as $S_{D}^t$,  $U_{D}^t$, and  $A_{D}$, respectively.
The linear operator in Eq.~(\ref{lineardiffusion}) is
\begin{align}
{\cal L} 
=
\frac{\pa^2}{\pa x^2},
\end{align}
and we obtain by partial integration
\begin{align}
\int_0^L \frac{\pa^2 u(x, t)}{\pa x^2} \overline{ v(x, t) } dx
=
J + \int_0^L u(x, t) \overline{ \frac{\pa^2 v(x, t)}{\pa x^2} } dx,
\end{align}
where $v(x, t) \in C^2$ and
\begin{align}
J = \left[ \frac{\pa u(x, t)}{\pa x} \overline{ v(x, t) } - u(x, t) \overline{ \frac{\pa v(x, t)}{\pa x} } \right]_0^L
\end{align}
is a bilinear concomitant. This $J$ vanishes by assuming homogeneous Neumann boundary conditions
\begin{align}
\frac{\partial v}{\partial x}(0, t) = \frac{\partial v}{\partial x}(L, t) = 0
\end{align}
also for $v(x, t)$. Thus, the problem is self-adjoint and if the weight function
$w(x) = w_{\lambda}(x)$ satisfies the adjoint eigenvalue equation
\begin{align}
{\cal L}^{*} w_{\lambda}(x) = \frac{\pa^2}{\pa x^2} w_{\lambda}(x) = \overline{\lambda} w_{\lambda}(x)
\label{adjointeq}
\end{align}
with the adjoint boundary conditions
\begin{align}
\frac{\partial w_\lambda}{\partial x}(0) = \frac{\partial w_\lambda}{\partial x}(L) = 0,
\label{adjointbc}
\end{align}
then
\begin{align}
\phi_{\lambda}[u] = \int_0^L u(x) \overline{w_{\lambda}(x)} dx
\end{align}
is an eigenfunctional of $A_{D}$ with the eigenvalue $\lambda$.

The eigenvalues of ${\cal L} = {\cal L}^*$ are given by real numbers
\begin{align}
\lambda_n = - \left( \frac{n \pi}{ L } \right)^2
\label{diffusion-neumann-eigenvalues}
\end{align}
and the weight functions $w_{\lambda}(x)$ can be taken as
\begin{align}
w_{n}(x) = \cos ( \sqrt{ | \lambda_n | } x ) = \cos \left( \frac{ n \pi }{ L } x \right),
\end{align}
where $n=1, 2, \cdots$. Note that the zero eigenvalue does not arise.

Therefore, the Fourier-cosine transform of $u$,
\begin{align}
\phi^D_{n}[u] = \int_0^L u(x) \cos\left( \frac{n \pi}{L} x \right) dx
\label{diffusionkoopman}
\end{align}
for $n=1, 2, \cdots$, gives the principal Koopman eigenfunctional of the linear diffusion equation~(\ref{lineardiffusion}) with the boundary conditions~(\ref{neumann}), where the Koopman eigenvalue $\lambda_n$ is given by Eq.~(\ref{diffusion-neumann-eigenvalues}).
Indeed, from Eq.~(\ref{adjoint0}), we have
\begin{align}
A_{D} \phi^D_{n}[u]
&= 
\int u(x) \left\{ \frac{\pa^2}{\pa x^2} \cos\left( \frac{n \pi}{L} x \right) \right\} dx 
\cr
&= - \left( \frac{n \pi}{L} \right)^2 \int u(x) \cos\left( \frac{n \pi}{L} x \right) dx
\cr
&=
\lambda_n \phi^D_{n}[u].
\end{align}

The inertial manifold ${\cal I}$ of the system is given by a set of functions satisfying $\phi^D_{n}[u] = 0$ for $n \geq 2$. From Eq.~(\ref{diffusionkoopman}), this is simply a function space spanned by $\cos ( \pi x / L)$. Similarly, the isostable characterizing asymptotic convergence to the stationary state is given by the level set of
$\phi^D_{1}[u] = \int_0^L u(x) \cos (\pi x / L) dx$,
i.e., by the 1st Fourier-cosine coefficient of the field variable.

In Fig.~\ref{fig1}, numerical results for the linear diffusion equation~(\ref{lineardiffusion}) with $L=20$ are presented.
Figure~\ref{fig1}(a) shows several snapshots of the field variable $u(x, t)$ during the relaxation to the stationary state $u_0(x) = 2$ (which is subtracted in the theory as mentioned before).
Figure~\ref{fig1}(b) plots the values of
\begin{align}
\phi_{n}^{D} [S_{D}^{t} u] = \int_0^L u(x, t) \cos \left( \frac{n \pi}{L} x \right) dx
\end{align}
as a function of time $t$ for $n=1, ..., 5$.
The exponential decay of the data points indicate that $\phi_{n}^D$ actually satisfies the eigenvalue equation
\begin{align}
\phi_{n}^D[ S_{D}^{t} u ] = U^t_{D} \phi_{n}^D[ u ] = e^{\lambda_n t}  \phi_{n}^D[ u ].
\end{align}
The Koopman eigenvalues evaluated numerically from the slopes of the data in Fig.~\ref{fig1}(b) are plotted in Fig.~\ref{fig4}, which agree well with the theoretical values given by Eq.~(\ref{diffusion-neumann-eigenvalues}), as expected.

\subsubsection{Dirichlet boundary conditions}

We next consider a linear diffusion equation on $[0, L]$,
\begin{align}
\frac{\pa}{\pa t} v(x, t) = \frac{\pa^2}{\pa x^2} v(x, t),
\label{lineardiffusion-d}
\end{align}
where $v(x, t)$ is the field variable, with inhomogeneous Dirichlet boundary conditions,
\begin{align}
v(0, t) = a, \quad v(L, t) = b,
\label{dirichlet0}
\end{align}
where $a$ and $b$ are real constants. In this case, Eq.~(\ref{lineardiffusion-d}) has a stationary solution
\begin{align}
v_0(x) = a + \frac{b-a}{L} x.
\label{diffusion-statioary}
\end{align}
We introduce a new field variable $u(x, t) = v(x, t) - v_0(x)$, which obeys
\begin{align}
\frac{\pa}{\pa t} u(x, t) = \frac{\pa^2}{\pa x^2} u(x, t)
\label{diffusiondirichlet}
\end{align}
with homogeneous Dirichlet boundary conditions,
\begin{align}
u(0, t) = u(L, t) = 0.
\label{dirichlet}
\end{align}
In this case, the bilinear concomitant $J$ vanishes by assuming homogeneous Dirichlet boundary conditions also for the adjoint problem. 
Thus, the problem is self-adjoint and the weight function $w(x) = w_{\lambda}(x)$ satisfies
\begin{align}
{\cal L}^{*} w_{\lambda}(x) = \frac{\pa^2}{\pa x^2} w_{\lambda}(x) = \overline{\lambda} w_{\lambda}(x)
\end{align}
and
\begin{align}
w(0) = w(L) = 0. 
\end{align}

The eigenvalues of ${\cal L} = {\cal L}^*$ in this case are the same as in the Neumann case and are given by 
\begin{align}
\lambda_n = - \left( \frac{n \pi}{ L } \right)^2,
\label{diffusion-dirichlet-eigenvalues}
\end{align}
and the weight functions $w_{\lambda}(x)$ are now given by
\begin{align}
w_{n}(x) = \sin ( \sqrt{ | \lambda_n | } x ) = \sin \left( \frac{ n \pi }{ L } x \right)
\end{align}
for $n=1, 2, \cdots$.

Therefore, the Fourier-sine transform of $u$,
\begin{align}
\phi^D_{n}[u] = \int_0^L u(x) \sin\left( \frac{n \pi}{L} x \right) dx,
\end{align}
are the principal Koopman eigenfunctionals of Eqs.~(\ref{diffusiondirichlet}) and~(\ref{dirichlet}).
In terms of the original variable $v$, this is given in the form
\begin{align}
\phi^D_{n}[v] = \int \{ v(x) - v_0(x) \} \sin\left( \frac{n \pi}{L} x \right) dx.
\label{Koopman-diffusion}
\end{align}

The inertial manifold ${\cal I}$ is now given by a function space spanned by $\sin ( \pi x / L)$, and the isostable is given by the level set of $\phi^D_{1}[u] = \int_0^L u(x) \sin (\pi x / L) dx$, i.e., by the 1st Fourier-sine coefficient of the field variable.

\subsubsection{Koopman modes}

As explained in Sec. II G, the 1st Koopman mode $d_1(x)$ of the linear diffusion equation are given by the eigenfunction $q_1(x)$ associated with the eigenvalue $\lambda_1$.
For Neumann boundary conditions, the field variable $u(x, t)$ converges to $0$ as $u(x, t) \simeq a_1 \exp( - \lambda_1 t ) \cos( \pi x / L )$ with $\lambda_1 = - ( \pi / L)^2$, so $d_1(x) = a_1 \cos( \pi x / L )$, where the constant $a_1$ is determined by the initial condition.
Indeed, we can observe in Fig.~\ref{fig1}(a) that the field variable $u$ at $t=30$ is approximately
a cosine function $\cos ( \pi x / L )$.
Similarly, for Dirichlet boundary conditions, $u(x, t)$ converges to $0$ as $u(x, t) \simeq a_1 \exp( -\lambda_1 t ) \sin( \pi x / L )$ and we have $d_1(x) = a_1 \sin( \pi x / L )$.
These functions represent, naturally, the slowest-decaying spatial modes of the linear diffusion equation with these boundary conditions.

\subsection{Burgers equation}

As the first example of the nonlinear PDE, we consider the viscous Burgers equation, which is exactly solvable.
The Koopman-operator and related analysis for the Burgers equation has been discussed by Kutz, Proctor, and Brunton~\cite{kutz2018koopman}, Page and Kerswell~\cite{page2018koopman}, Peitz and Klus~\cite{peitz2019koopman},
Wilson and Djouadi~\cite{wilson2020adaptive},
and by Balabane, Mendez, and Najem~\cite{balabane}.
In Refs.~\cite{kutz2018koopman,page2018koopman}, the spectrum and Koopman modes have been analyzed by using DMD and kernel methods.
In Ref.~\cite{peitz2019koopman}, model reduction and control of the Burgers equation is discussed and in Ref.~\cite{wilson2020adaptive}, isostable reduction for control of the Burgers equation is performed.
In Ref.~\cite{balabane}, rigorous mathematical analysis on the Koopman-operator of the Burgers equation is performed.
Here, we present the results on the Koopman eigenvalues and eigenfunctionals of the Burgers equation in the context of conjugacy with the linear diffusion equation.

The viscous Burgers equation can be expressed as
\begin{align}
\frac{\pa}{\pa t} z(x, t) = - z(x, t) \frac{\pa}{\pa x} z(x, t) + \frac{\pa^2}{\pa x^2} z(x, t)
\label{burgers}
\end{align}
after rescaling, where $z(x, t)$ is the field variable on $[0, L]$. We assume homogeneous Dirichlet boundary conditions,
\begin{align}
z(0, t) = z(L, t) = 0.
\label{dirichlet-burgers}
\end{align}
In this case, the field variable $z(x, t)$ converges to a stable stationary state $z_0(x) = 0$ as $t \to \infty$ after transient.
We denote the flow, Koopman operator, and generator of the Koopman operator of this system as $S_{B}^t$,  $U_{B}^t$, and  $A_{B}$, respectively.

\begin{figure}[t]
\includegraphics[width=0.85\hsize]{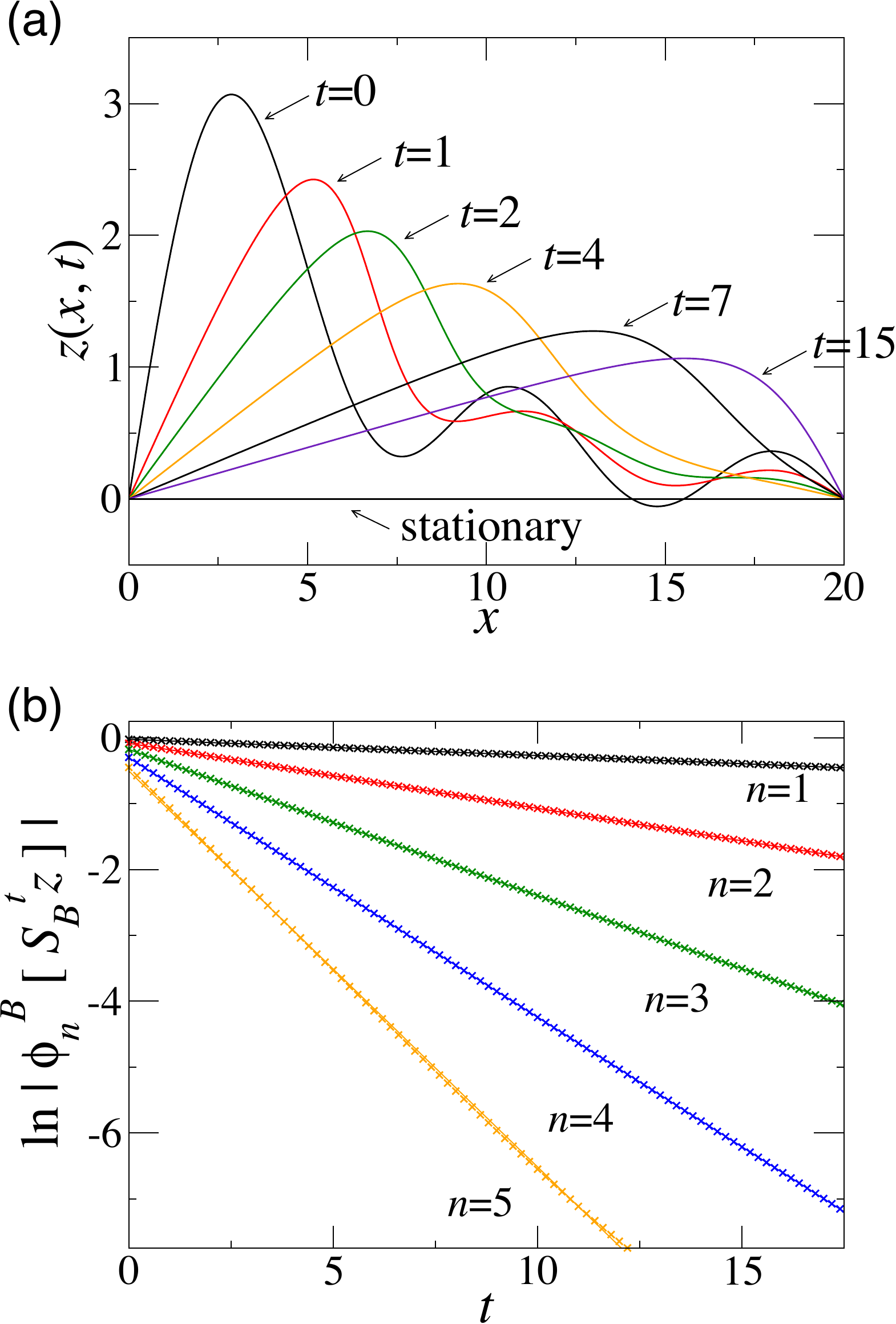}
\caption{Burgers equation with Dirichlet boundary conditions. (a) Relaxation of the field variable $z(x, t)$ to a uniform stationary state. (b) Exponential decay of the Koopman eigenfunctionals $\phi^B_{n}[z(x,t)]$ for $n=1, \cdots, 5$. The initial state is $z(x, t) =  \sin( \pi x/L ) + 0.9  \sin( \pi x/L ) + 0.8 \sin( 3 \pi x / L ) + 0.7 \sin( 4 \pi x / L ) + 0.6 \sin( 5 \pi x/L )$. 
}
\label{fig2}
\end{figure}

It is well known that the Burgers equation and linear diffusion equation are related via the Hopf-Cole transformation~\cite{hopf1950partial,cole1951quasi},
\begin{align}
z(x, t) = - 2 \frac{\pa}{\pa x} \ln v(x, t) = - 2 \frac{1}{v(x, t)} \frac{\pa}{\pa x} v(x, t),
\label{hopfcole}
\end{align}
\begin{align}
v(x, t) = c(t) \exp \left( - \frac{1}{2} \int_0^x z(y, t) dy \right),
\label{hopfcolev}
\end{align}
where the coefficient $c(t)$ should be chosen appropriately so that $v(x, t)$ obeys the diffusion equation.
Namely, if $z(x, t)$ satisfies Eqs.~(\ref{burgers}) and~(\ref{dirichlet-burgers}), then $v(x, t)$ satisfies the linear diffusion equation
\begin{align}
\frac{\pa}{\pa t} v(x, t) = \frac{\pa^2}{\pa x^2} v(x, t)
\label{lineardiffusion2}
\end{align}
with the Neumann boundary conditions
\begin{align}
\frac{\partial v}{\partial x}(0, t) = \frac{\partial v}{\partial x}(L, t) = 0,
\label{neumann-diffusion}
\end{align}
and if $v(x, t)$ satisfies Eqs.~(\ref{lineardiffusion2}) and~(\ref{neumann-diffusion}), then $z(x, t)$ satisfies the Burgers equation~(\ref{burgers}) with the Dirichlet boundary conditions~(\ref{dirichlet-burgers}).
Thus, the Hopf-Cole transformation gives the mapping $\Phi$ of conjugacy in Eq.~(\ref{conj}).

Substituting Eq.~(\ref{hopfcolev}) into Eq.~(\ref{lineardiffusion2}), it turns out that $c(t)$ should satisfy
\begin{align}
\frac{d}{dt} \ln c(t) = \frac{1}{2} \left( - \left. \frac{\partial z(x, t)}{\partial x}\right|_{x=0} + \frac{1}{2} z(0, t)^2 \right).
\end{align}
In the present case with the Dirichlet boundary conditions for $z(x, t)$, the above equation is satisfied by~\cite{page2018koopman} 
\begin{align}
c(t) = \left( \int_0^L \exp \left( - \frac{1}{2} \int_0^x z(y, t) dy \right) dx \right)^{-1}.
\end{align}
The solution $v(x, t)$ of Eq.~(\ref{lineardiffusion2}) with Eq.~(\ref{neumann-diffusion}) approaches a constant $v_{\infty}$ as $t \to \infty$, 
\begin{align}
v_{\infty} = \lim_{t \to \infty} v(x, t) = \lim_{t \to \infty} c(t) = \frac{1}{L},
\end{align}
and by introducing a new field variable
\begin{align}
u(x, t) = v(x, t) - v_{\infty},
\end{align}
$u(x, t)$ satisfies Eqs.~(\ref{lineardiffusion}) and Eq.~(\ref{neumann}) and $u(x, t) \to 0$ as $t \to \infty$, for which Eq.~(\ref{diffusionkoopman}) gives the Koopman eigenfunctional.
Therefore, we find that the functional
\begin{align}
\phi_{n}^B[z] =
\int_0^L \left\{ c[z] \exp \left( - \frac{1}{2} \int_0^x z(y) dy \right) - v_{\infty} \right\}
\cr
\times
\cos\left( \frac{n \pi}{L} x \right) dx
\label{burgerskoopmaneigenfn}
\end{align}
with
\begin{align}
c[z] = \left( \int_0^L \exp \left( - \frac{1}{2} \int_0^x z(y) dy \right) dx \right)^{-1}
\label{burgerskoopmancoef}
\end{align}
gives the principal Koopman eigenfunctional of the Burgers equation~(\ref{burgers}) with the Dirichlet boundary conditions~(\ref{dirichlet-burgers}) satisfying
\begin{align}
\phi_{n}^B[ S_{B}^{t} z ] = U^t_B \phi_{n}^B[ z ] = e^{\lambda_n t}  \phi_{n}^B[ z ],
\label{burgers-exponentialdecay}
\end{align}
where we represented the coefficient $c(t)$ in Eq.~(\ref{hopfcolev}) as a functional $c[z]$ of the field variable $z$.
Because of the conjugacy, the associated eigenvalues $\lambda_n$ are the same as those for the linear diffusion equation with Neumann boundary conditions,
\begin{align}
\lambda_n = - \left( \frac{n \pi}{ L } \right)^2
\quad
 (n=1, 2, \cdots).
 \label{burgers-theory-exponents}
\end{align}
The inertial manifold ${\cal I}$ for the Burgers equation is given by a set of functions $z$ satisfying $\phi^B_{n}[z] = 0$ for $n \geq 2$, and the isostables are given by the level sets of $\phi^B_{\lambda_1}[z]$.

Figure~\ref{fig2}(a) shows several snapshots of the field variable $z(x, t) = S_{B}^{t} z(x, 0)$ of the Burgers equation with $L=20$ during relaxation to the stationary solution obtained by direct numerical simulations.
It can be seen that the peak of $z(x, t)$ gradually moves to the right due to advection and then decays to zero due to viscosity.
Figure~\ref{fig2}(b) plots the values of $\phi_{n}^B[ S_{B}^{t} z ]$ ($n=1, 2, ..., 5$) given by Eqs.~(\ref{burgerskoopmaneigenfn}) and (\ref{burgerskoopmancoef}) with respect to $t$ during the relaxation, which exhibit clear exponential decay despite the nonlinear evolution of $z(x, t)$, indicating that Eq.~(\ref{burgers-exponentialdecay}) is actually satisfied.
The Koopman eigenvalues evaluated numerically from the slopes of the data points in Fig.~\ref{fig2}(b) are plotted in Fig.~\ref{fig4}, which agree well with the theoretical values, Eq.~(\ref{burgers-theory-exponents}). 

In Ref.~\cite{page2018koopman}, Page and Kerswell reported that the Koopman eigenvalues are highly degenerate, namely, the multiplicity of the Koopman eigenvalue increases quickly with $n$. This is because the products of principal Koopman eigenfunctionals are also Koopman eigenfunctionals and because the eigenvalues, Eq.~(\ref{burgers-theory-exponents}), have quite a simple quadratic dependence on $n$.
In such cases, due to the same reason, the expansion of the field variable in Eq.~(\ref{expansion}) will be composed of a number of different Koopman eigenfunctionals with the same growth rate.

In Ref.~\cite{ROBERTS2000187}, a modified Burgers equation, which exhibits a pitchfork bifurcation, is considered and an explicit expression of the center manifold near the bifurcation point is derived.
Though not discussed in the present study, the Koopman-operator approach can also be developed for PDE's possessing a center manifold and, when spectral expansions exist, we can use it for uniquely constructing the center manifold in a similar way to the inertial manifold discussed in this study (see Ref.~\cite{mezic2017koopman}, Sec. 3). It would then be insightful to compare the results obtained by the standard center manifold theory and that obtained by the newly developed Koopman operator approach.

\begin{figure}[t]
\includegraphics[width=0.85\hsize]{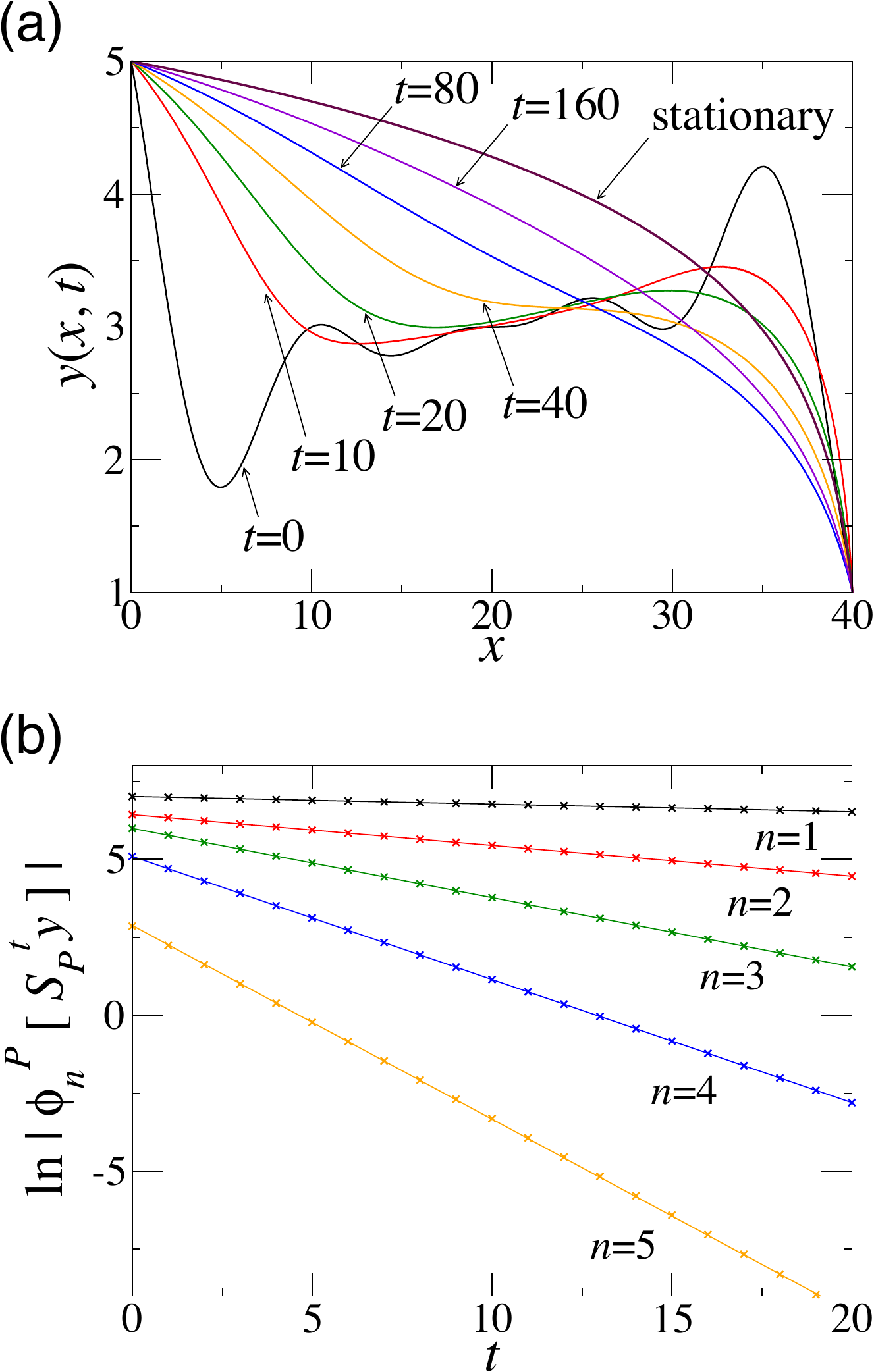}
\caption{Nonlinear phase-diffusion equation with inhomogeneous Dirichlet boundary conditions. (a) Relaxation of the system to a uniform state. (b) Exponential decay of the Koopman eigenfunctionals $\phi^P_{n}[y(x,t)]$ for $n=1, \cdots, 5$. The boundary values are $q=5$ and $p=1$, and the initial state is given by $y(x, 0) = 3 + 0.5 \{ \cos( \pi x/L ) - \sin( 2\pi x/L ) + \cos( 3\pi x/L ) - \sin( 4\pi x/L ) + \cos( 5\pi x/L ) - \sin( 6\pi x/L) + \cos( 7\pi x/L ) - \sin( 8\pi x/L) \}$.}
\label{fig3}
\end{figure}

\subsection{Nonlinear phase-diffusion equation}

As the second example of the nonlinear PDE, we consider a ``nonlinear phase-diffusion equation'' on $[0, L]$, which describes the dynamics of the phase field of self-oscillatory media, such as the propagation and collision of phase waves~\cite{kuramoto2012chemical}. It is given by
\begin{align}
\frac{\pa}{\pa t} y(x, t) = \frac{\pa^2}{\pa x^2} y(x, t) + \left( \frac{\pa }{\pa x} y(x, t) \right)^2
\label{phasediffusion}
\end{align}
after rescaling, where $y(x, t)$ is the field variable representing local phase of an oscillatory medium and natural frequency of the medium is subtracted without loss of generality. 
The above equation also gives the deterministic part of the Kardar-Parisi-Zhang equation describing stochastic growth of rough interfaces~\cite{kardar1986dynamic}.
We assume inhomogeneous Dirichlet boundary conditions,
\begin{align}
y(0, t) = q, \quad y(L, t) = p.
\end{align}
We denote the flow, Koopman operator, and generator of the Koopman operator of this system as $S_{P}^t$,  $U_{P}^t$, and  $A_{P}$, respectively.

As shown in Ref.~\cite{kuramoto2012chemical}, the phase-diffusion equation~(\ref{phasediffusion}) is closely related to the Burgers equation and, by introducing a new field variable
\begin{align}
v(x, t) = \exp y(x, t) ,
\label{phasediffusionconjug}
\end{align}
it transforms into the linear diffusion equation~(\ref{lineardiffusion-d}) with the inhomogeneous Dirichlet boundary conditions, Eq.~(\ref{dirichlet0}) with $a = e^{\beta q}$ and $b = e^{\beta p}$.
The stationary solution Eq.~(\ref{diffusion-statioary}) for the linear diffusion equation corresponds to a stationary solution
\begin{align}
y_0(x)
&= 
\frac{1}{\beta} \ln \left( e^{\beta q} + \frac{e^{\beta p} - e^{\beta q}}{L} x \right) 
\end{align}
of the nonlinear phase-diffusion equation~(\ref{phasediffusion}).

Thus, from Eq.~(\ref{Koopman-diffusion}) for the principal Koopman eigenfunctional of the linear diffusion equation with Dirichlet boundary conditions and 
Eq.~(\ref{phasediffusionconjug}) for the conjugacy, we obtain
\begin{align}
\phi_{n}^P[y] =
\int \{ \exp y(x) - \exp y_0(x) \}
\sin\left( \frac{n \pi}{L} x \right) dx
\label{phasekoopman}
\end{align}
as the Koopman eigenfunctional of $U^t_P$ and $A_P$ of the nonlinear phase-diffusion equation associated with eigenvalue
\begin{align}
\lambda_n = - \left( \frac{n \pi}{ L } \right)^2 \quad (n=1, 2, \cdots),
\end{align}
satisfying
\begin{align}
\phi_{n}^P[ S_{P}^{t} y] = U^t_P \phi_{n}^P[ y ] = e^{\lambda_n t}  \phi_{n}^P[ y ].
\label{exponentialdecayphase}
\end{align}
As before, the inertial manifold ${\cal I}$ for the nonlinear phase-diffusion equation is given by a set of functions satisfying $\phi^P_{n}[y] = 0$ for $n \geq 2$, and the isostables are given by the level sets of $\phi^P_{\lambda_1}[y]$.

Figure~\ref{fig3}(a) shows several snapshots of the field variable $y(x, t) = S_{P}^{t} y(x, 0)$ during relaxation to the stationary solution $y_0(x)$ of the nonlinear phase-diffusion equation with $L=40$ obtained by direct numerical simulations. The phase field converges to a stationary pattern determined by the boundary conditions.
Figure~\ref{fig3}(b) plots {the values of $\phi_{n}^P[ S_{P}^{t} y ]$ ($n=1, 2, ..., 5$) with respect to $t$}, which again exhibit clear exponential decay despite nonlinear evolution of $y(x, t)$ {and indicate that Eq.~(\ref{exponentialdecayphase}) is actually satisfied}.
The Koopman eigenvalues evaluated numerically from the slopes of the data points in Fig.~\ref{fig3}(b) are plotted in Fig.~\ref{fig4}, which agree well with the theoretical values. Because of the conjugacy, the eigenvalues of the nonlinear phase-diffusion equation coincide with those of the diffusion equation as well as the Burgers equation.

\begin{figure}[thbp]
\includegraphics[width=0.85\hsize]{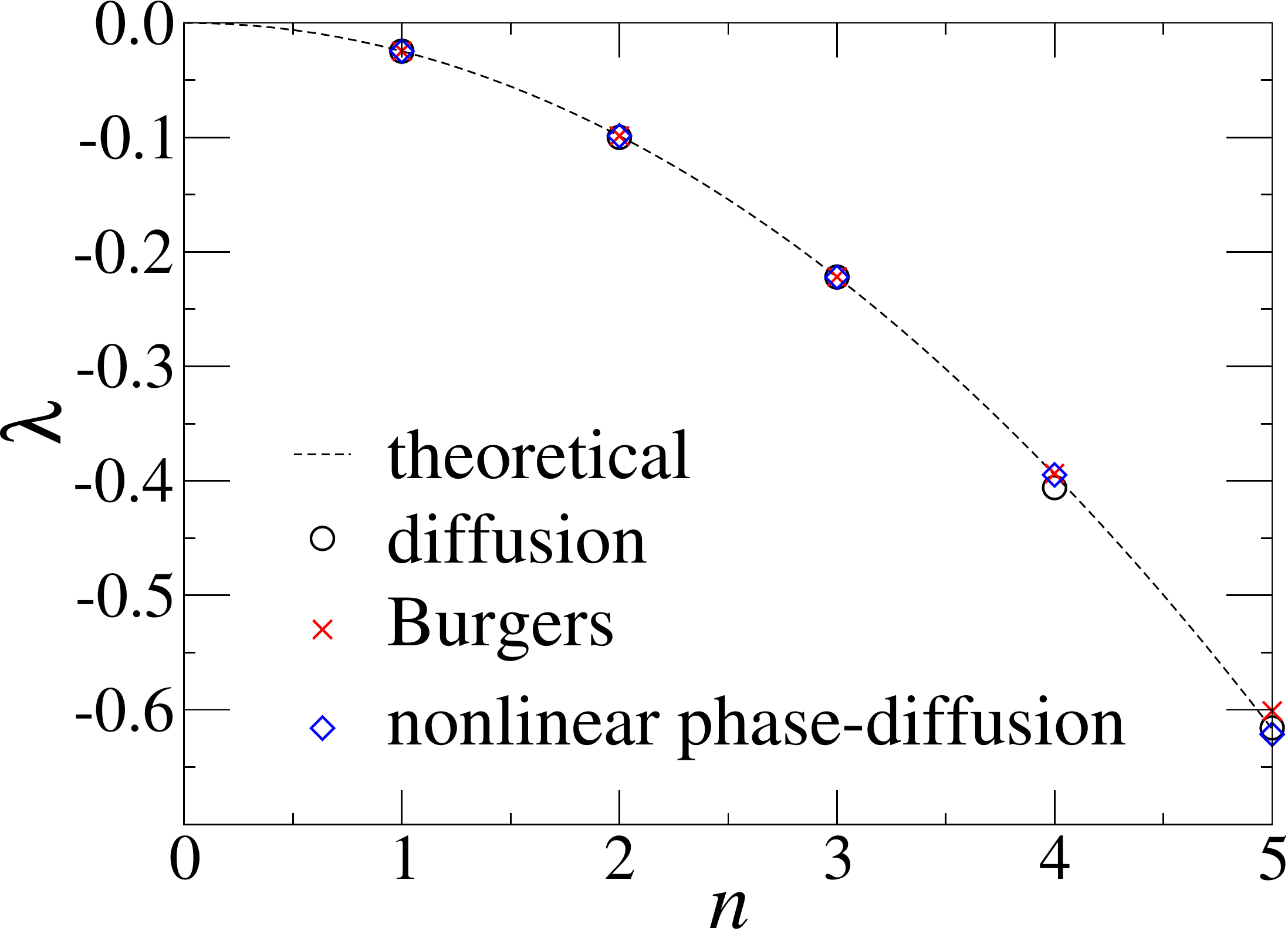}
\caption{Koopman eigenvalues for the linear diffusion, Burgers, and nonlinear phase-diffusion equations. Numerical results are compared with theoretical values.}
\label{fig4}
\end{figure}

\section{Summary}

We have shown that the spectral Koopman-operator formalism can formally be generalized to a PDE describing relaxation of the system state to a stationary state by introducing the concept of the Koopman eigenfunctionals.
Using exactly solvable examples of linear and nonlinear PDE's (treated via conjugacy), we have illustrated that the system can be decomposed into independent Koopman modes, and the dynamics of the system can be described by a set of linear equations for them. 
We have also provided the general definition of isostables for a class of PDE's and the inertial manifolds by using level sets of Koopman eigenfunctionals.
These results have been illustrated for three PDE's that are mutually conjugate, namely, the diffusion equation, Burgers equation, and nonlinear phase-diffusion equation.

The Koopman operator analysis for PDE's discussed in this paper provides a general framework for linearizing nonlinear dynamical systems described by PDE's in the basin of attraction of the stationary state, which will be useful for their analysis and control. The same framework can be further generalized to include limit-cycling PDE's, which is closely related to the phase or phase-amplitude reduction methods for PDE's~\cite{nakao2014phase,wilson2016isostable,nakao2018sice,nakao2020springer}.

We have used the conjugacy relation between Burgers and nonlinear phase-diffusion equations with the linear diffusion equation to derive their Koopman eigenvalues and eigenfunctionals.
It is generally expected that the original nonlinear system is conjugate to its linearized system around an exponentially stable, hyperbolic stationary state of the PDE.
Such a conjugacy relation has been proven, e.g., in the context of the Hartman-Grobman theorem for reaction-diffusion equations in a neighborhood of a stationary solution \cite{lu1991hartman}.

Finally, though we considered only solvable examples for which the Koopman eigenfunctionals can be explicitly obtained, it is generally difficult to obtain the Koopman eigenfunctionals explicitly. However, we should still be able to derive rigorous theoretical results and develop approximate or numerical methods for PDE's on the basis of the general formulation of the Koopman operator analysis. For example, we can use statistical or machine-learning methods to approximate the dominant Koopman eigenfunctional for given PDE's, and then use them to analyze their dynamics and develop methods for control and optimization.

\acknowledgments{
H. N. thanks JSPS KAKENHI JP17H03279, JP18H03287, and JST CREST JPMJCR1913 for financial support.
I. M. acknowledges support from ARO-MURI W911NF-17-1-0306.
}

\section{Appendix}

  \renewcommand{\theequation}{A\arabic{equation}}
  % redefine the command that creates the equation no.
  \setcounter{equation}{0}  % reset counter 

In this appendix, the definition of the functional derivatives and derivation of the expansion of analytic observables are informally given using physics notations. Mathematical theories on Banach calculus can be found in Refs.~\cite{dieudonne2011foundations,zeidler2012applied,Teschl}. 

\subsection{Functional derivatives}

The functional derivative $\delta g[u] / \delta u(x)$ of a smooth functional $g[u]$ by $u(x)$ is defined by a G\^ateaux derivative as
\begin{align}
\left. \frac{dg[u+ {\epsilon'} \eta]}{d{\epsilon'}}\right|_{{\epsilon'}=0}
= \int \frac{\delta g[u]}{\delta u(x)} \eta(x) dx
\label{fndrv}
\end{align}
where $\eta(x)$ is a test function. The left-hand side of Eq.~(\ref{fndrv}) is explicitly given by 
\begin{align}
\left. \frac{dg[u + {\epsilon'} \eta]}{d{\epsilon'}}\right|_{{\epsilon'}=0}
&=
\lim_{{\epsilon'} \to 0} \frac{g[u + {\epsilon'} \eta] - g[u]}{{\epsilon'}}.
\end{align}
The generator $A$ of the Koopman operator $U^{\tau}$ can be derived using this expression as
\begin{align}
A g[u]
&= \lim_{\tau \to 0} \frac{ U^{\tau} g[u] - g[u] } {\tau}
= \lim_{\tau \to 0} \frac{ g[S^{\tau} u] - g[u] } {\tau}
\cr
&= \lim_{\tau \to 0} \frac{ g[u + \tau {\cal F}\{u(x)\} + O(\tau^2)] - g[u] } {\tau}
\cr
&= \left. \frac{dg[u+\tau {\cal F}\{ u(x) \}]}{d\tau}\right|_{\tau=0}
\cr
&= \int_0^L \frac{\delta g[u]}{\delta u(x)} {\cal F}\{ u(x) \} dx
\end{align}
where we used
\begin{align}
S^{\tau} u(x) = u(x) + \tau {\cal F}\{u(x)\} + O(\tau^2) 
\end{align}
for sufficiently small $\tau$.

Similarly, higher-order functional derivatives of a smooth functional $g[u]$ by $u(x)$ is
defined by Taylor expansion of $g[u + \epsilon \eta]$ in $\epsilon$ around $\epsilon = 0$,
\begin{align}
g[ u + \epsilon \eta ]
=& g[ u ] + \left. \frac{dg[u+{\epsilon'} \eta]}{d{\epsilon'}} \right|_{{\epsilon'}=0} \epsilon
\cr
&+ \frac{1}{2} \left. \frac{d^2 g[u+{\epsilon'} \eta]}{d{\epsilon'}^2} \right|_{{\epsilon'}=0} \epsilon^2 + \cdots
\cr
=&
\sum_{n=0}^{\infty} \frac{1}{n!} \left. \frac{d^n g[u+{\epsilon'} \eta]}{d{\epsilon'}^n} \right|_{{\epsilon'}=0} \epsilon^n,
\end{align}
where the $n$-th expansion coefficient gives the $n$-th order functional derivative of $g[u]$ evaluated at $u(x)$,
\begin{align}
&
\left. \frac{d^n g[u+{\epsilon'} \eta]}{d{\epsilon'}^n} \right|_{{\epsilon'}=0} = \int \eta(x_1) \eta(x_2) \cdots \eta(x_n) 
\cr
&
\quad
\times
\frac{ \delta^n g[u] }{ \delta u(x_1) \delta u(x_2) \cdots \delta u(x_n) } dx_1 dx_2 \cdots dx_n.
\end{align}
Note that
\begin{align}
\left. \frac{d^0 g[u+{\epsilon'} \eta]}{d{\epsilon'}^0} \right|_{{\epsilon'}=0} = g[u]
\end{align}
when $n=0$.

Setting $\epsilon = 1$ in the above expression, the functional Taylor expansion of $g[u + \eta]$ around ${u}$ is given by
\begin{align}
g[u+\eta] 
=&
g[u] + \int \frac{\delta g[u]}{\delta u(x_1)} \eta(x_1) dx_1
\cr
&+ \frac{1}{2!} \int \frac{\delta^2 g[u]}{\delta u(x_1) \delta u(x_2) } \eta(x_1) \eta(x_2) dx_1 dx_2 + \cdots
\cr
=& g[u] + \sum_{n=1}^{\infty} \frac{1}{n!} \int \frac{ \delta^n g[u] }{ \delta u(x_1) \delta u(x_2) \cdots \delta u(x_n) }
\cr
&\times
\eta(x_1) \eta(x_2) \cdots \eta(x_n) dx_1 dx_2 \cdots dx_n.
\end{align}
For a finite Taylor series of functionals, see e.g.~\cite{zeidler2012applied}.

\subsection{Derivation of the expansion of $g[u]$}

We consider the system Eq.~(\ref{systemexpanded}) with an exponentially stable stationary state $u_0(x) = 0$.
We assume that the linearized operator ${\cal L}$ of ${\cal F}$ around $u_0(x) = 0$ has a discrete spectrum of eigenvalues \{$\lambda_j$\}, where $\lim_{j \to \infty} \mbox{Re\ }\lambda_j = -\infty$ and this is the only accumulation point. 
We denote by $\phi_{\lambda_j}[ u ]$ a Koopman eigenfunctional associated with $\lambda_{j}$.

We assume that the observation functional $g[u]$ is analytic and can be expanded in functional Taylor series around $u=0$ as
\begin{align}
g[u]
=& \sum_{n=0}^{\infty} \frac{1}{n!} \int \frac{ \delta^n g[u] }{ \delta u(x_1) \delta u(x_2) \cdots \delta u(x_n) }
\cr
\times
& u(x_1) u(x_2) \cdots u(x_n) dx_1 dx_2 \cdots dx_n
\label{functionalTaylor}
\end{align}
where the functional derivatives are evaluated at $u(x) = u_0(x) = 0$.

By using the Koopman eigenfunctionals, we introduce new variables from the field variable $u(x)$ as
\begin{align}
y_i = \phi_{\lambda_i}[ u ] \quad (i=1, 2, ... ),
\end{align}
each of which satisfies $d{y}_i/dt = \lambda_i y_i$ from Eq.~(\ref{generatoreigen}).
When $u(x) = u_0(x) = 0$, we have $y_i = \phi_{\lambda_i}[0] = 0$ for all $i$, because $u_0(x) = 0$ is exponentially stable and therefore $\mbox{Re} \lambda_i < 0$.
We assume that the field variable $u(x)$ can be expressed inversely using these new variables as
\begin{align}
u(x) = V_x (y_1, y_2, ...), \quad (0 \leq x \leq L),
\end{align}
where $V_x$ is a $x$-dependent function satisfying $V_x(0, 0, ...) = u_0(x) = 0$.

We assume that the field variable $u(x) = V_x(y_1, y_2, ...)$ can be expanded around $(y_1, y_2, ...) = (0, 0, ...)$ as
\begin{align}
u(x)
=& \sum_j \frac{\partial V_x}{\partial y_j} y_j + \frac{1}{2!} \sum_{j,k} \frac{\partial^2 V_x}{\partial y_j \partial y_k} y_j y_k
\cr
&+ \frac{1}{3!} \sum_{j,k,l} \frac{\partial^3 V_x}{\partial y_j \partial y_k \partial y_l} y_j y_k y_l + \cdots,
\end{align}
where the partial derivatives of $V_x(y_1, ..., y_N)$ are evaluated at $(y_1, y_2, ...) = (0, 0, ...)$ and each of $j, k$, and $l$ runs from $1$ to $\infty$.
Plugging this expression into the functional Taylor expansion of $g[u]$ in Eq.~(\ref{functionalTaylor}), we obtain
\begin{align}
&
g[u]
= g[0] + \int dx_1 \frac{\delta g[u]}{\delta u(x_1)} \Big\{
\sum_j \frac{ \partial V_{x_1}}{ \partial y_j } y_j
\cr
&\quad
+ \frac{1}{2!} \sum_{j,k} \frac{\partial^2 V_{x_1}}{\partial y_j \partial y_k} y_j y_k
+ \frac{1}{3!} \sum_{j,k,l} \frac{\partial^3 V_{x_1}}{\partial y_j \partial y_k \partial y_l} y_j y_k y_l + \cdots
\Big\}
\cr
&+
\frac{1}{2!} \int dx_1 dx_2 \frac{\delta^2 g[u]}{\delta u(x_1) \delta u(x_2)} 
\Big\{ \sum_{j,k}
\frac{\partial V_{x_1}}{\partial y_j} \frac{\partial V_{x_2}}{\partial y_k} y_j y_k
\cr
&\quad
+\frac{1}{2!} \sum_{j,k,l}
\left(
\frac{\partial V_{x_1}}{\partial y_j}
\frac{\partial^2 V_{x_2}}{\partial y_k \partial y_l}
+
\frac{\partial V_{x_2}}{\partial y_j}
\frac{\partial^2 V_{x_1}}{\partial y_k \partial y_l}
\right) y_j y_k y_l
+ \cdots \Big\}
\cr
&+
\frac{1}{3!} \int dx_1 dx_2 dx_3 \frac{\delta^3 g[u]}{\delta u(x_1) \delta u(x_2) \delta u(x_3)}
\cr
&\quad
\times
\Big\{ \sum_{j, k, l} \frac{\partial V_{x_1}}{\partial y_j} \frac{\partial V_{x_2}}{\partial y_k} \frac{\partial V_{x_3}}{\partial y_l} y_j y_k y_l + \cdots \Big\}
+ \cdots
\end{align}

Thus, the observation functional $g[u]$ can be expanded using the Koopman eigenfunctionals as
\begin{align}
&
g[u] = g[0]
+ \sum_j \phi_{\lambda_j}[u] \int \frac{\delta g[u]}{\delta u(x_1)} \frac{ \partial V_{x_1}}{ \partial y_j } dx_1
\cr
&+ \frac{1}{2!} \sum_{j,k} \phi_{\lambda_j}[u] \phi_{\lambda_k}[u] \Big\{ \int \frac{\delta g[u]}{\delta u(x_1)} \frac{\partial^2 V_{x_1}}{\partial y_j \partial y_k} d{x_1}
\cr
&\quad
+ \int \frac{\delta^2 g[u]}{\delta u(x_1) \delta u(x_2)} \frac{\partial V_{x_1}}{\partial y_j} \frac{\partial V_{x_2}}{\partial y_k} dx_1 dx_2 \Big\}
\cr
&+ \frac{1}{3!} \sum_{j,k,l} \phi_{\lambda_j}[u] \phi_{\lambda_k}[u] \phi_{\lambda_l}[u]
\cr
&\quad
\times \Big\{
\int \frac{\delta g[u]}{\delta u(x_1)} \frac{\partial^3 V_{x_1}}{\partial y_j \partial y_k \partial y_l} d{x_1}
+
\frac{3}{2}\int \frac{\delta^2 g[u]}{\delta u(x_1) \delta u(x_2)}
\cr
&\quad\quad
\times \left( \frac{\partial V_{x_1}}{\partial y_j}
\frac{\partial^2 V_{x_2}}{\partial y_k \partial y_l}
+
\frac{\partial V_{x_2}}{\partial y_j}
\frac{\partial^2 V_{x_1}}{\partial y_k \partial y_l}
 \right) d{x_1} d{x_2} 
\cr
&\quad\quad
+ \int \frac{\delta^3 g[u]}{\delta u(x_1) \delta u(x_2) \delta u(x_3)} \frac{\partial V_{x_1}}{\partial y_j} \frac{\partial V_{x_2}}{\partial y_k} \frac{\partial V_{x_3}}{\partial y_l} d{x_1} d{x_2} d{x_3}
\Big\}
+ \cdots 
\cr
&=
g[0] + \sum_{n=1}^{\infty} \frac{1}{n!}
\sum_{j_1, j_2, ..., j_n}
c_{j_1, j_2, ..., j_n}
\phi_{\lambda_{j_1}}[u] \phi_{\lambda_{j_2}}[u] \cdots \phi_{\lambda_{j_n}}[u],
\cr
\end{align}
where $c_{j_1, j_2, ..., j_n}$ are the expansion coefficients.
The first three of them are given by
\begin{align}
c_{j_1} =& \int \frac{\delta g[u]}{\delta u(x_1)} \frac{ \partial V_{x_1}}{ \partial y_{j_1} } d{x_1},
\end{align}
\begin{align}
c_{j_1, j_2} =& \int \frac{\delta g[u]}{\delta u(x_1)} \frac{\partial^2 V_{x_1}}{\partial y_{j_1} \partial y_{j_2}} d{x_1}
\cr
&+ \int \frac{\delta^2 g[u]}{\delta u(x_1) \delta u(x_2)} \frac{\partial V_{x_1}}{\partial y_{j_1}} \frac{\partial V_{x_2}}{\partial y_{j_2}} dx_1 dx_2,
\quad
\end{align}
\begin{align}
c_{j_1, j_2, j_3} =&
\int \frac{\delta g[u]}{\delta u(x_1)} \frac{\partial^3 V_{x_1}}{\partial y_{j_1} \partial y_{j_2} \partial y_{j_3}} d{x_1}
\cr
&+
\frac{3}{2}\int \frac{\delta^2 g[u]}{\delta u(x_1) \delta u(x_2)}
\Big( \frac{\partial V_{x_1}}{\partial y_{j_1}}
\frac{\partial^2 V_{x_2}}{\partial y_{j_2} \partial y_{j_3}}
\cr
&\quad\quad
+
\frac{\partial V_{x_2}}{\partial y_{j_1}}
\frac{\partial^2 V_{x_1}}{\partial y_{j_2} \partial y_{j_3}}
\Big) d{x_1} d{x_2} 
\cr
&+
\int \frac{\delta^3 g[u]}{\delta u(x_1) \delta u(x_2) \delta u(x_3)}
\cr
&\quad\quad
\times
\frac{\partial V_{x_1}}{\partial y_{j_1}} \frac{\partial V_{x_2}}{\partial y_{j_2}} \frac{\partial V_{x_3}}{\partial y_{j_3}} d{x_1} d{x_2} d{x_3},
\quad
\end{align}
and can further be calculated in a similar way.

\newpage


\begin{thebibliography}{50}%
\makeatletter
\providecommand \@ifxundefined [1]{%
 \@ifx{#1\undefined}
}%
\providecommand \@ifnum [1]{%
 \ifnum #1\expandafter \@firstoftwo
 \else \expandafter \@secondoftwo
 \fi
}%
\providecommand \@ifx [1]{%
 \ifx #1\expandafter \@firstoftwo
 \else \expandafter \@secondoftwo
 \fi
}%
\providecommand \natexlab [1]{#1}%
\providecommand \enquote  [1]{``#1''}%
\providecommand \bibnamefont  [1]{#1}%
\providecommand \bibfnamefont [1]{#1}%
\providecommand \citenamefont [1]{#1}%
\providecommand \href@noop [0]{\@secondoftwo}%
\providecommand \href [0]{\begingroup \@sanitize@url \@href}%
\providecommand \@href[1]{\@@startlink{#1}\@@href}%
\providecommand \@@href[1]{\endgroup#1\@@endlink}%
\providecommand \@sanitize@url [0]{\catcode `\\12\catcode `\$12\catcode
  `\&12\catcode `\#12\catcode `\^12\catcode `\_12\catcode `\%12\relax}%
\providecommand \@@startlink[1]{}%
\providecommand \@@endlink[0]{}%
\providecommand \url  [0]{\begingroup\@sanitize@url \@url }%
\providecommand \@url [1]{\endgroup\@href {#1}{\urlprefix }}%
\providecommand \urlprefix  [0]{URL }%
\providecommand \Eprint [0]{\href }%
\providecommand \doibase [0]{http://dx.doi.org/}%
\providecommand \selectlanguage [0]{\@gobble}%
\providecommand \bibinfo  [0]{\@secondoftwo}%
\providecommand \bibfield  [0]{\@secondoftwo}%
\providecommand \translation [1]{[#1]}%
\providecommand \BibitemOpen [0]{}%
\providecommand \bibitemStop [0]{}%
\providecommand \bibitemNoStop [0]{.\EOS\space}%
\providecommand \EOS [0]{\spacefactor3000\relax}%
\providecommand \BibitemShut  [1]{\csname bibitem#1\endcsname}%
\let\auto@bib@innerbib\@empty
%</preamble>
\bibitem [{\citenamefont {Mezi\'c}()}]{MezicBook}%
  \BibitemOpen
  \bibfield  {author} {\bibinfo {author} {\bibfnamefont {I.}~\bibnamefont
  {Mezi\'c}},\ }\href@noop {} {\emph {\bibinfo {title} {Spectral {K}oopman
  Operator Methods in Dynamical Systems}}}\BibitemShut {NoStop}%
\bibitem [{\citenamefont {Mauroy}, \citenamefont {Susuki},\ and\ \citenamefont
  {Mezi{\'c}}(2020)}]{mauroy2020introduction}%
  \BibitemOpen
  \bibfield  {author} {\bibinfo {author} {\bibfnamefont {A.}~\bibnamefont
  {Mauroy}}, \bibinfo {author} {\bibfnamefont {Y.}~\bibnamefont {Susuki}}, \
  and\ \bibinfo {author} {\bibfnamefont {I.}~\bibnamefont {Mezi{\'c}}},\
  }\href@noop {} {\emph {\bibinfo {title} {The {K}oopman Operator in Systems
  and Control}}}\ (\bibinfo  {publisher} {Springer},\ \bibinfo {year}
  {2020})\BibitemShut {NoStop}%
\bibitem [{\citenamefont {Budi{\v{s}}i{\'c}}, \citenamefont {Mohr},\ and\
  \citenamefont {Mezi{\'c}}(2012)}]{budivsic2012applied}%
  \BibitemOpen
  \bibfield  {author} {\bibinfo {author} {\bibfnamefont {M.}~\bibnamefont
  {Budi{\v{s}}i{\'c}}}, \bibinfo {author} {\bibfnamefont {R.}~\bibnamefont
  {Mohr}}, \ and\ \bibinfo {author} {\bibfnamefont {I.}~\bibnamefont
  {Mezi{\'c}}},\ }\bibfield  {title} {\enquote {\bibinfo {title} {Applied
  {K}oopmanism},}\ }\href@noop {} {\bibfield  {journal} {\bibinfo  {journal}
  {Chaos: An Interdisciplinary Journal of Nonlinear Science}\ }\textbf
  {\bibinfo {volume} {22}},\ \bibinfo {pages} {047510} (\bibinfo {year}
  {2012})}\BibitemShut {NoStop}%
\bibitem [{\citenamefont {Mezi{\'c}}(2013)}]{mezic2013analysis}%
  \BibitemOpen
  \bibfield  {author} {\bibinfo {author} {\bibfnamefont {I.}~\bibnamefont
  {Mezi{\'c}}},\ }\bibfield  {title} {\enquote {\bibinfo {title} {Analysis of
  fluid flows via spectral properties of the {K}oopman operator},}\ }\href@noop
  {} {\bibfield  {journal} {\bibinfo  {journal} {Annual Review of Fluid
  Mechanics}\ }\textbf {\bibinfo {volume} {45}},\ \bibinfo {pages} {357--378}
  (\bibinfo {year} {2013})}\BibitemShut {NoStop}%
\bibitem [{\citenamefont {Lan}\ and\ \citenamefont
  {Mezi{\'c}}(2013)}]{lan2013linearization}%
  \BibitemOpen
  \bibfield  {author} {\bibinfo {author} {\bibfnamefont {Y.}~\bibnamefont
  {Lan}}\ and\ \bibinfo {author} {\bibfnamefont {I.}~\bibnamefont
  {Mezi{\'c}}},\ }\bibfield  {title} {\enquote {\bibinfo {title} {Linearization
  in the large of nonlinear systems and {K}oopman operator spectrum},}\
  }\href@noop {} {\bibfield  {journal} {\bibinfo  {journal} {Physica D:
  Nonlinear Phenomena}\ }\textbf {\bibinfo {volume} {242}},\ \bibinfo {pages}
  {42--53} (\bibinfo {year} {2013})}\BibitemShut {NoStop}%
\bibitem [{\citenamefont {Mezi{\'c}}(2019)}]{mezic2017koopman}%
  \BibitemOpen
  \bibfield  {author} {\bibinfo {author} {\bibfnamefont {I.}~\bibnamefont
  {Mezi{\'c}}},\ }\bibfield  {title} {\enquote {\bibinfo {title} {Spectrum of
  the koopman operator, spectral expansions in functional spaces, and
  state-space geometry},}\ }\href@noop {} {\bibfield  {journal} {\bibinfo
  {journal} {Journal of Nonlinear Science}\ }\textbf {\bibinfo {volume} {Online
  First}},\ \bibinfo {pages} {https://doi.org/10.1007/s00332--019--09598--5,
  1--55} (\bibinfo {year} {2019})}\BibitemShut {NoStop}%
\bibitem [{\citenamefont {Korda}, \citenamefont {Putinar},\ and\ \citenamefont
  {Mezi{\'c}}(2020)}]{korda2017data}%
  \BibitemOpen
  \bibfield  {author} {\bibinfo {author} {\bibfnamefont {M.}~\bibnamefont
  {Korda}}, \bibinfo {author} {\bibfnamefont {M.}~\bibnamefont {Putinar}}, \
  and\ \bibinfo {author} {\bibfnamefont {I.}~\bibnamefont {Mezi{\'c}}},\
  }\bibfield  {title} {\enquote {\bibinfo {title} {Data-driven spectral
  analysis of the koopman operator},}\ }\href@noop {} {\bibfield  {journal}
  {\bibinfo  {journal} {Applied and Computational Harmonic Analysis}\ }\textbf
  {\bibinfo {volume} {48}},\ \bibinfo {pages} {599--629} (\bibinfo {year}
  {2020})}\BibitemShut {NoStop}%
\bibitem [{\citenamefont {Surana}\ and\ \citenamefont
  {Banaszuk}(2016)}]{surana2016linear}%
  \BibitemOpen
  \bibfield  {author} {\bibinfo {author} {\bibfnamefont {A.}~\bibnamefont
  {Surana}}\ and\ \bibinfo {author} {\bibfnamefont {A.}~\bibnamefont
  {Banaszuk}},\ }\bibfield  {title} {\enquote {\bibinfo {title} {Linear
  observer synthesis for nonlinear systems using {K}oopman operator
  framework},}\ }\href@noop {} {\bibfield  {journal} {\bibinfo  {journal}
  {IFAC-PapersOnLine}\ }\textbf {\bibinfo {volume} {49}},\ \bibinfo {pages}
  {716--723} (\bibinfo {year} {2016})}\BibitemShut {NoStop}%
\bibitem [{\citenamefont {Gaspard}(2005)}]{gaspard2005chaos}%
  \BibitemOpen
  \bibfield  {author} {\bibinfo {author} {\bibfnamefont {P.}~\bibnamefont
  {Gaspard}},\ }\href@noop {} {\emph {\bibinfo {title} {Chaos, scattering and
  statistical mechanics}}},\ Vol.~\bibinfo {volume} {9}\ (\bibinfo  {publisher}
  {Cambridge University Press},\ \bibinfo {year} {2005})\BibitemShut {NoStop}%
\bibitem [{\citenamefont {Koopman}(1931)}]{koopman1931hamiltonian}%
  \BibitemOpen
  \bibfield  {author} {\bibinfo {author} {\bibfnamefont {B.~O.}\ \bibnamefont
  {Koopman}},\ }\bibfield  {title} {\enquote {\bibinfo {title} {Hamiltonian
  systems and transformation in {H}ilbert space},}\ }\href@noop {} {\bibfield
  {journal} {\bibinfo  {journal} {Proceedings of the national academy of
  sciences of the united states of america}\ }\textbf {\bibinfo {volume}
  {17}},\ \bibinfo {pages} {315} (\bibinfo {year} {1931})}\BibitemShut
  {NoStop}%
\bibitem [{\citenamefont {von Neumann}(1932)}]{vonneumann1932}%
  \BibitemOpen
  \bibfield  {author} {\bibinfo {author} {\bibfnamefont {J.}~\bibnamefont {von
  Neumann}},\ }\bibfield  {title} {\enquote {\bibinfo {title} {Zur
  operatorenmethode in der klassischen mechanik},}\ }\href@noop {} {\bibfield
  {journal} {\bibinfo  {journal} {Annals of Mathematics}\ }\textbf {\bibinfo
  {volume} {33}},\ \bibinfo {pages} {587--642} (\bibinfo {year}
  {1932})}\BibitemShut {NoStop}%
\bibitem [{\citenamefont {Winfree}(1980)}]{winfree2001geometry}%
  \BibitemOpen
  \bibfield  {author} {\bibinfo {author} {\bibfnamefont {A.~T.}\ \bibnamefont
  {Winfree}},\ }\href@noop {} {\emph {\bibinfo {title} {The geometry of
  biological time}}}\ (\bibinfo  {publisher} {Springer, New York},\ \bibinfo
  {year} {1980})\BibitemShut {NoStop}%
\bibitem [{\citenamefont {Kuramoto}(1984)}]{kuramoto2012chemical}%
  \BibitemOpen
  \bibfield  {author} {\bibinfo {author} {\bibfnamefont {Y.}~\bibnamefont
  {Kuramoto}},\ }\href@noop {} {\emph {\bibinfo {title} {Chemical oscillations,
  waves, and turbulence}}}\ (\bibinfo  {publisher} {Springer, New York},\
  \bibinfo {year} {1984})\BibitemShut {NoStop}%
\bibitem [{\citenamefont {Hoppensteadt}\ and\ \citenamefont
  {Izhikevich}(1997)}]{hoppensteadt2012weakly}%
  \BibitemOpen
  \bibfield  {author} {\bibinfo {author} {\bibfnamefont {F.~C.}\ \bibnamefont
  {Hoppensteadt}}\ and\ \bibinfo {author} {\bibfnamefont {E.~M.}\ \bibnamefont
  {Izhikevich}},\ }\href@noop {} {\emph {\bibinfo {title} {Weakly connected
  neural networks}}}\ (\bibinfo  {publisher} {Springer, New York},\ \bibinfo
  {year} {1997})\BibitemShut {NoStop}%
\bibitem [{\citenamefont {Ermentrout}\ and\ \citenamefont
  {Terman}(2010)}]{ermentrout2010mathematical}%
  \BibitemOpen
  \bibfield  {author} {\bibinfo {author} {\bibfnamefont {G.~B.}\ \bibnamefont
  {Ermentrout}}\ and\ \bibinfo {author} {\bibfnamefont {D.~H.}\ \bibnamefont
  {Terman}},\ }\href@noop {} {\emph {\bibinfo {title} {Mathematical foundations
  of neuroscience}}},\ Vol.~\bibinfo {volume} {35}\ (\bibinfo  {publisher}
  {Springer, New York},\ \bibinfo {year} {2010})\BibitemShut {NoStop}%
\bibitem [{\citenamefont {Nakao}(2016)}]{nakao2016phase}%
  \BibitemOpen
  \bibfield  {author} {\bibinfo {author} {\bibfnamefont {H.}~\bibnamefont
  {Nakao}},\ }\bibfield  {title} {\enquote {\bibinfo {title} {Phase reduction
  approach to synchronisation of nonlinear oscillators},}\ }\href@noop {}
  {\bibfield  {journal} {\bibinfo  {journal} {Contemporary Physics}\ }\textbf
  {\bibinfo {volume} {57}},\ \bibinfo {pages} {188--214} (\bibinfo {year}
  {2016})}\BibitemShut {NoStop}%
\bibitem [{\citenamefont {Mauroy}, \citenamefont {Mezi{\'c}},\ and\
  \citenamefont {Moehlis}(2013)}]{mauroy2013isostables}%
  \BibitemOpen
  \bibfield  {author} {\bibinfo {author} {\bibfnamefont {A.}~\bibnamefont
  {Mauroy}}, \bibinfo {author} {\bibfnamefont {I.}~\bibnamefont {Mezi{\'c}}}, \
  and\ \bibinfo {author} {\bibfnamefont {J.}~\bibnamefont {Moehlis}},\
  }\bibfield  {title} {\enquote {\bibinfo {title} {Isostables, isochrons, and
  {K}oopman spectrum for the action--angle representation of stable fixed point
  dynamics},}\ }\href@noop {} {\bibfield  {journal} {\bibinfo  {journal}
  {Physica D: Nonlinear Phenomena}\ }\textbf {\bibinfo {volume} {261}},\
  \bibinfo {pages} {19--30} (\bibinfo {year} {2013})}\BibitemShut {NoStop}%
\bibitem [{\citenamefont {Mauroy}(2014)}]{mauroy2014converging}%
  \BibitemOpen
  \bibfield  {author} {\bibinfo {author} {\bibfnamefont {A.}~\bibnamefont
  {Mauroy}},\ }\bibfield  {title} {\enquote {\bibinfo {title} {Converging to
  and escaping from the global equilibrium: Isostables and optimal control},}\
  }in\ \href@noop {} {\emph {\bibinfo {booktitle} {Decision and Control (CDC),
  2014 IEEE 53rd Annual Conference on}}}\ (\bibinfo {organization} {IEEE},\
  \bibinfo {year} {2014})\ pp.\ \bibinfo {pages} {5888--5893}\BibitemShut
  {NoStop}%
\bibitem [{\citenamefont {Wilson}\ and\ \citenamefont
  {Moehlis}(2015)}]{wilson2015extending}%
  \BibitemOpen
  \bibfield  {author} {\bibinfo {author} {\bibfnamefont {D.}~\bibnamefont
  {Wilson}}\ and\ \bibinfo {author} {\bibfnamefont {J.}~\bibnamefont
  {Moehlis}},\ }\bibfield  {title} {\enquote {\bibinfo {title} {Extending phase
  reduction to excitable media: theory and applications},}\ }\href@noop {}
  {\bibfield  {journal} {\bibinfo  {journal} {SIAM Review}\ }\textbf {\bibinfo
  {volume} {57}},\ \bibinfo {pages} {201--222} (\bibinfo {year}
  {2015})}\BibitemShut {NoStop}%
\bibitem [{\citenamefont {Wilson}\ and\ \citenamefont
  {Moehlis}(2016)}]{wilson2016isostable}%
  \BibitemOpen
  \bibfield  {author} {\bibinfo {author} {\bibfnamefont {D.}~\bibnamefont
  {Wilson}}\ and\ \bibinfo {author} {\bibfnamefont {J.}~\bibnamefont
  {Moehlis}},\ }\bibfield  {title} {\enquote {\bibinfo {title} {Isostable
  reduction with applications to time-dependent partial differential
  equations},}\ }\href@noop {} {\bibfield  {journal} {\bibinfo  {journal}
  {Physical Review E}\ }\textbf {\bibinfo {volume} {94}},\ \bibinfo {pages}
  {012211} (\bibinfo {year} {2016})}\BibitemShut {NoStop}%
\bibitem [{\citenamefont {Mezi{\'c}}(2005)}]{mezic2005spectral}%
  \BibitemOpen
  \bibfield  {author} {\bibinfo {author} {\bibfnamefont {I.}~\bibnamefont
  {Mezi{\'c}}},\ }\bibfield  {title} {\enquote {\bibinfo {title} {Spectral
  properties of dynamical systems, model reduction and decompositions},}\
  }\href@noop {} {\bibfield  {journal} {\bibinfo  {journal} {Nonlinear
  Dynamics}\ }\textbf {\bibinfo {volume} {41}},\ \bibinfo {pages} {309--325}
  (\bibinfo {year} {2005})}\BibitemShut {NoStop}%
\bibitem [{\citenamefont {Page}\ and\ \citenamefont
  {Kerswell}(2018)}]{page2018koopman}%
  \BibitemOpen
  \bibfield  {author} {\bibinfo {author} {\bibfnamefont {J.}~\bibnamefont
  {Page}}\ and\ \bibinfo {author} {\bibfnamefont {R.~R.}\ \bibnamefont
  {Kerswell}},\ }\bibfield  {title} {\enquote {\bibinfo {title} {Koopman
  analysis of burgers equation},}\ }\href {\doibase
  10.1103/PhysRevFluids.3.071901} {\bibfield  {journal} {\bibinfo  {journal}
  {Phys. Rev. Fluids}\ }\textbf {\bibinfo {volume} {3}},\ \bibinfo {pages}
  {071901} (\bibinfo {year} {2018})}\BibitemShut {NoStop}%
\bibitem [{\citenamefont {Kutz}, \citenamefont {Proctor},\ and\ \citenamefont
  {Brunton}(2018)}]{kutz2018koopman}%
  \BibitemOpen
  \bibfield  {author} {\bibinfo {author} {\bibfnamefont {J.~N.}\ \bibnamefont
  {Kutz}}, \bibinfo {author} {\bibfnamefont {J.~L.}\ \bibnamefont {Proctor}}, \
  and\ \bibinfo {author} {\bibfnamefont {S.~L.}\ \bibnamefont {Brunton}},\
  }\bibfield  {title} {\enquote {\bibinfo {title} {Applied koopman theory for
  partial differential equations and data-driven modeling of spatio-temporal
  systems},}\ }\href@noop {} {\bibfield  {journal} {\bibinfo  {journal}
  {Complexity}\ }\textbf {\bibinfo {volume} {2018}},\ \bibinfo {pages}
  {6010634} (\bibinfo {year} {2018})}\BibitemShut {NoStop}%
\bibitem [{\citenamefont {Nakao}, \citenamefont {Yanagita},\ and\ \citenamefont
  {Kawamura}(2014)}]{nakao2014phase}%
  \BibitemOpen
  \bibfield  {author} {\bibinfo {author} {\bibfnamefont {H.}~\bibnamefont
  {Nakao}}, \bibinfo {author} {\bibfnamefont {T.}~\bibnamefont {Yanagita}}, \
  and\ \bibinfo {author} {\bibfnamefont {Y.}~\bibnamefont {Kawamura}},\
  }\bibfield  {title} {\enquote {\bibinfo {title} {Phase-reduction approach to
  synchronization of spatiotemporal rhythms in reaction-diffusion systems},}\
  }\href@noop {} {\bibfield  {journal} {\bibinfo  {journal} {Physical review
  X}\ }\textbf {\bibinfo {volume} {4}},\ \bibinfo {pages} {021032} (\bibinfo
  {year} {2014})}\BibitemShut {NoStop}%
\bibitem [{\citenamefont {Teschl}(2019)}]{Teschl}%
  \BibitemOpen
  \bibfield  {author} {\bibinfo {author} {\bibfnamefont {G.}~\bibnamefont
  {Teschl}},\ }\href@noop {} {\emph {\bibinfo {title} {Topics in real and
  functional analysis}}}\ (\bibinfo  {publisher} {Amer. Math. Soc.},\ \bibinfo
  {year} {2019})\BibitemShut {NoStop}%
\bibitem [{\citenamefont {Bollt}\ \emph {et~al.}(2018)\citenamefont {Bollt},
  \citenamefont {Li}, \citenamefont {Dietrich},\ and\ \citenamefont
  {Kevrekidis}}]{bollt2018matching}%
  \BibitemOpen
  \bibfield  {author} {\bibinfo {author} {\bibfnamefont {E.~M.}\ \bibnamefont
  {Bollt}}, \bibinfo {author} {\bibfnamefont {Q.}~\bibnamefont {Li}}, \bibinfo
  {author} {\bibfnamefont {F.}~\bibnamefont {Dietrich}}, \ and\ \bibinfo
  {author} {\bibfnamefont {I.}~\bibnamefont {Kevrekidis}},\ }\bibfield  {title}
  {\enquote {\bibinfo {title} {On matching, and even rectifying, dynamical
  systems through koopman operator eigenfunctions},}\ }\href@noop {} {\bibfield
   {journal} {\bibinfo  {journal} {SIAM Journal on Applied Dynamical Systems}\
  }\textbf {\bibinfo {volume} {17}},\ \bibinfo {pages} {1925--1960} (\bibinfo
  {year} {2018})}\BibitemShut {NoStop}%
\bibitem [{\citenamefont {Williams}\ \emph {et~al.}(2015)\citenamefont
  {Williams}, \citenamefont {Rowley}, \citenamefont {Mezi{\'c}},\ and\
  \citenamefont {Kevrekidis}}]{williams2015data}%
  \BibitemOpen
  \bibfield  {author} {\bibinfo {author} {\bibfnamefont {M.~O.}\ \bibnamefont
  {Williams}}, \bibinfo {author} {\bibfnamefont {C.~W.}\ \bibnamefont
  {Rowley}}, \bibinfo {author} {\bibfnamefont {I.}~\bibnamefont {Mezi{\'c}}}, \
  and\ \bibinfo {author} {\bibfnamefont {I.~G.}\ \bibnamefont {Kevrekidis}},\
  }\bibfield  {title} {\enquote {\bibinfo {title} {Data fusion via intrinsic
  dynamic variables: An application of data-driven {K}oopman spectral
  analysis},}\ }\href@noop {} {\bibfield  {journal} {\bibinfo  {journal} {EPL
  (Europhysics Letters)}\ }\textbf {\bibinfo {volume} {109}},\ \bibinfo {pages}
  {40007} (\bibinfo {year} {2015})}\BibitemShut {NoStop}%
\bibitem [{\citenamefont {Foias}, \citenamefont {Sell},\ and\ \citenamefont
  {Temam}(1988)}]{foias1988inertial}%
  \BibitemOpen
  \bibfield  {author} {\bibinfo {author} {\bibfnamefont {C.}~\bibnamefont
  {Foias}}, \bibinfo {author} {\bibfnamefont {G.~R.}\ \bibnamefont {Sell}}, \
  and\ \bibinfo {author} {\bibfnamefont {R.}~\bibnamefont {Temam}},\ }\bibfield
   {title} {\enquote {\bibinfo {title} {Inertial manifolds for nonlinear
  evolutionary equations},}\ }\href@noop {} {\bibfield  {journal} {\bibinfo
  {journal} {Journal of Differential Equations}\ }\textbf {\bibinfo {volume}
  {73}},\ \bibinfo {pages} {309--353} (\bibinfo {year} {1988})}\BibitemShut
  {NoStop}%
\bibitem [{\citenamefont {Foias}\ \emph {et~al.}(1988)\citenamefont {Foias},
  \citenamefont {Jolly}, \citenamefont {Kevrekidis}, \citenamefont {Sell},\
  and\ \citenamefont {Titi}}]{foias1988computation}%
  \BibitemOpen
  \bibfield  {author} {\bibinfo {author} {\bibfnamefont {C.}~\bibnamefont
  {Foias}}, \bibinfo {author} {\bibfnamefont {M.}~\bibnamefont {Jolly}},
  \bibinfo {author} {\bibfnamefont {I.}~\bibnamefont {Kevrekidis}}, \bibinfo
  {author} {\bibfnamefont {G.}~\bibnamefont {Sell}}, \ and\ \bibinfo {author}
  {\bibfnamefont {E.}~\bibnamefont {Titi}},\ }\bibfield  {title} {\enquote
  {\bibinfo {title} {On the computation of inertial manifolds},}\ }\href@noop
  {} {\bibfield  {journal} {\bibinfo  {journal} {Physics Letters A}\ }\textbf
  {\bibinfo {volume} {131}},\ \bibinfo {pages} {433--436} (\bibinfo {year}
  {1988})}\BibitemShut {NoStop}%
\bibitem [{\citenamefont {Robinson}(2001)}]{robinson2001infinite}%
  \BibitemOpen
  \bibfield  {author} {\bibinfo {author} {\bibfnamefont {J.~C.}\ \bibnamefont
  {Robinson}},\ }\href@noop {} {\emph {\bibinfo {title} {Infinite-dimensional
  dynamical systems: an introduction to dissipative parabolic PDEs and the
  theory of global attractors}}},\ Vol.~\bibinfo {volume} {28}\ (\bibinfo
  {publisher} {Cambridge University Press},\ \bibinfo {year}
  {2001})\BibitemShut {NoStop}%
\bibitem [{\citenamefont {Jolly}, \citenamefont {Rosa},\ and\ \citenamefont
  {Temam}(2001)}]{jolly2001accurate}%
  \BibitemOpen
  \bibfield  {author} {\bibinfo {author} {\bibfnamefont {M.~S.}\ \bibnamefont
  {Jolly}}, \bibinfo {author} {\bibfnamefont {R.}~\bibnamefont {Rosa}}, \ and\
  \bibinfo {author} {\bibfnamefont {R.}~\bibnamefont {Temam}},\ }\bibfield
  {title} {\enquote {\bibinfo {title} {Accurate computations on inertial
  manifolds},}\ }\href@noop {} {\bibfield  {journal} {\bibinfo  {journal} {SIAM
  Journal on Scientific Computing}\ }\textbf {\bibinfo {volume} {22}},\
  \bibinfo {pages} {2216--2238} (\bibinfo {year} {2001})}\BibitemShut {NoStop}%
\bibitem [{\citenamefont {Schmid}(2010)}]{schmid2010dynamic}%
  \BibitemOpen
  \bibfield  {author} {\bibinfo {author} {\bibfnamefont {P.~J.}\ \bibnamefont
  {Schmid}},\ }\bibfield  {title} {\enquote {\bibinfo {title} {Dynamic mode
  decomposition of numerical and experimental data},}\ }\href@noop {}
  {\bibfield  {journal} {\bibinfo  {journal} {Journal of fluid mechanics}\
  }\textbf {\bibinfo {volume} {656}},\ \bibinfo {pages} {5--28} (\bibinfo
  {year} {2010})}\BibitemShut {NoStop}%
\bibitem [{\citenamefont {Rowley}\ \emph {et~al.}(2009)\citenamefont {Rowley},
  \citenamefont {Mezi{\'c}}, \citenamefont {Bagheri}, \citenamefont
  {Schlatter},\ and\ \citenamefont {Henningson}}]{rowley2009spectral}%
  \BibitemOpen
  \bibfield  {author} {\bibinfo {author} {\bibfnamefont {C.~W.}\ \bibnamefont
  {Rowley}}, \bibinfo {author} {\bibfnamefont {I.}~\bibnamefont {Mezi{\'c}}},
  \bibinfo {author} {\bibfnamefont {S.}~\bibnamefont {Bagheri}}, \bibinfo
  {author} {\bibfnamefont {P.}~\bibnamefont {Schlatter}}, \ and\ \bibinfo
  {author} {\bibfnamefont {D.~S.}\ \bibnamefont {Henningson}},\ }\bibfield
  {title} {\enquote {\bibinfo {title} {Spectral analysis of nonlinear flows},}\
  }\href@noop {} {\bibfield  {journal} {\bibinfo  {journal} {Journal of Fluid
  Mechanics}\ }\textbf {\bibinfo {volume} {641}},\ \bibinfo {pages} {115--127}
  (\bibinfo {year} {2009})}\BibitemShut {NoStop}%
\bibitem [{\citenamefont {Keener}(1988)}]{keener1988principles}%
  \BibitemOpen
  \bibfield  {author} {\bibinfo {author} {\bibfnamefont {J.~P.}\ \bibnamefont
  {Keener}},\ }\href@noop {} {\emph {\bibinfo {title} {Principles of applied
  mathematics}}}\ (\bibinfo  {publisher} {Addison-Wesley},\ \bibinfo {year}
  {1988})\BibitemShut {NoStop}%
\bibitem [{Note1()}]{Note1}%
  \BibitemOpen
  \bibinfo {note} {This result can easily be generalized for periodic
  attractors, where instead of claiming the number of eigenvalues is finite in
  a subset of left half plane, we can claim the number is finite in a
  rectangular subset of the complex plane that includes the imaginary axis and
  $0$ on its right boundary.}\BibitemShut {Stop}%
\bibitem [{\citenamefont {Wilson}\ and\ \citenamefont
  {Djouadi}(2020)}]{wilson2020adaptive}%
  \BibitemOpen
  \bibfield  {author} {\bibinfo {author} {\bibfnamefont {D.}~\bibnamefont
  {Wilson}}\ and\ \bibinfo {author} {\bibfnamefont {S.~M.}\ \bibnamefont
  {Djouadi}},\ }\bibfield  {title} {\enquote {\bibinfo {title} {Adaptive
  isostable reduction of nonlinear pdes with time varying parameters},}\
  }\href@noop {} {\bibfield  {journal} {\bibinfo  {journal} {IEEE Control
  Systems Letters}\ }\textbf {\bibinfo {volume} {5}},\ \bibinfo {pages}
  {187--192} (\bibinfo {year} {2020})}\BibitemShut {NoStop}%
\bibitem [{\citenamefont {Mauroy}\ and\ \citenamefont
  {Mezi{\'c}}(2016)}]{mauroy2016global}%
  \BibitemOpen
  \bibfield  {author} {\bibinfo {author} {\bibfnamefont {A.}~\bibnamefont
  {Mauroy}}\ and\ \bibinfo {author} {\bibfnamefont {I.}~\bibnamefont
  {Mezi{\'c}}},\ }\bibfield  {title} {\enquote {\bibinfo {title} {Global
  stability analysis using the eigenfunctions of the {K}oopman operator},}\
  }\href@noop {} {\bibfield  {journal} {\bibinfo  {journal} {IEEE Transactions
  on Automatic Control}\ }\textbf {\bibinfo {volume} {61}},\ \bibinfo {pages}
  {3356--3369} (\bibinfo {year} {2016})}\BibitemShut {NoStop}%
\bibitem [{\citenamefont {Shirasaka}, \citenamefont {Kurebayashi},\ and\
  \citenamefont {Nakao}(2017)}]{shirasaka2017phase}%
  \BibitemOpen
  \bibfield  {author} {\bibinfo {author} {\bibfnamefont {S.}~\bibnamefont
  {Shirasaka}}, \bibinfo {author} {\bibfnamefont {W.}~\bibnamefont
  {Kurebayashi}}, \ and\ \bibinfo {author} {\bibfnamefont {H.}~\bibnamefont
  {Nakao}},\ }\bibfield  {title} {\enquote {\bibinfo {title} {Phase-amplitude
  reduction of transient dynamics far from attractors for limit-cycling
  systems},}\ }\href@noop {} {\bibfield  {journal} {\bibinfo  {journal} {Chaos:
  An Interdisciplinary Journal of Nonlinear Science}\ }\textbf {\bibinfo
  {volume} {27}},\ \bibinfo {pages} {023119} (\bibinfo {year}
  {2017})}\BibitemShut {NoStop}%
\bibitem [{\citenamefont {Shirasaka}, \citenamefont {Kurebayashi},\ and\
  \citenamefont {Nakao}(2020)}]{shirasaka2020phase}%
  \BibitemOpen
  \bibfield  {author} {\bibinfo {author} {\bibfnamefont {S.}~\bibnamefont
  {Shirasaka}}, \bibinfo {author} {\bibfnamefont {W.}~\bibnamefont
  {Kurebayashi}}, \ and\ \bibinfo {author} {\bibfnamefont {H.}~\bibnamefont
  {Nakao}},\ }\bibfield  {title} {\enquote {\bibinfo {title} {Phase-amplitude
  reduction of limit cycling systems},}\ }in\ \href@noop {} {\emph {\bibinfo
  {booktitle} {The Koopman Operator in Systems and Control}}}\ (\bibinfo
  {publisher} {Springer},\ \bibinfo {year} {2020})\ pp.\ \bibinfo {pages}
  {383--417}\BibitemShut {NoStop}%
\bibitem [{\citenamefont {Peitz}\ and\ \citenamefont
  {Klus}(2019)}]{peitz2019koopman}%
  \BibitemOpen
  \bibfield  {author} {\bibinfo {author} {\bibfnamefont {S.}~\bibnamefont
  {Peitz}}\ and\ \bibinfo {author} {\bibfnamefont {S.}~\bibnamefont {Klus}},\
  }\bibfield  {title} {\enquote {\bibinfo {title} {Koopman operator-based model
  reduction for switched-system control of pdes},}\ }\href@noop {} {\bibfield
  {journal} {\bibinfo  {journal} {Automatica}\ }\textbf {\bibinfo {volume}
  {106}},\ \bibinfo {pages} {184--191} (\bibinfo {year} {2019})}\BibitemShut
  {NoStop}%
\bibitem [{\citenamefont {Balabane}, \citenamefont {Mendez},\ and\
  \citenamefont {Najem}()}]{balabane}%
  \BibitemOpen
  \bibfield  {author} {\bibinfo {author} {\bibfnamefont {M.}~\bibnamefont
  {Balabane}}, \bibinfo {author} {\bibfnamefont {M.~A.}\ \bibnamefont
  {Mendez}}, \ and\ \bibinfo {author} {\bibfnamefont {S.}~\bibnamefont
  {Najem}},\ }\bibfield  {title} {\enquote {\bibinfo {title} {On koopman
  operators for {B}urgers equation},}\ }\href@noop {} {\bibinfo  {journal}
  {arXiv:2007.01218v1}\ }\BibitemShut {NoStop}%
\bibitem [{\citenamefont {Hopf}(1950)}]{hopf1950partial}%
  \BibitemOpen
\bibfield  {journal} {  }\bibfield  {author} {\bibinfo {author} {\bibfnamefont
  {E.}~\bibnamefont {Hopf}},\ }\bibfield  {title} {\enquote {\bibinfo {title}
  {The partial differential equation ut+ uux= $\mu$xx},}\ }\href@noop {}
  {\bibfield  {journal} {\bibinfo  {journal} {Communications on Pure and
  Applied mathematics}\ }\textbf {\bibinfo {volume} {3}},\ \bibinfo {pages}
  {201--230} (\bibinfo {year} {1950})}\BibitemShut {NoStop}%
\bibitem [{\citenamefont {Cole}(1951)}]{cole1951quasi}%
  \BibitemOpen
  \bibfield  {author} {\bibinfo {author} {\bibfnamefont {J.~D.}\ \bibnamefont
  {Cole}},\ }\bibfield  {title} {\enquote {\bibinfo {title} {On a quasi-linear
  parabolic equation occurring in aerodynamics},}\ }\href@noop {} {\bibfield
  {journal} {\bibinfo  {journal} {Quarterly of applied mathematics}\ }\textbf
  {\bibinfo {volume} {9}},\ \bibinfo {pages} {225--236} (\bibinfo {year}
  {1951})}\BibitemShut {NoStop}%
\bibitem [{\citenamefont {Roberts}(2000)}]{ROBERTS2000187}%
  \BibitemOpen
  \bibfield  {author} {\bibinfo {author} {\bibfnamefont {A.}~\bibnamefont
  {Roberts}},\ }\bibfield  {title} {\enquote {\bibinfo {title} {Computer
  algebra derives correct initial conditions for low-dimensional dynamical
  models},}\ }\href {\doibase https://doi.org/10.1016/S0010-4655(99)00494-4}
  {\bibfield  {journal} {\bibinfo  {journal} {Computer Physics Communications}\
  }\textbf {\bibinfo {volume} {126}},\ \bibinfo {pages} {187 -- 206} (\bibinfo
  {year} {2000})}\BibitemShut {NoStop}%
\bibitem [{\citenamefont {Kardar}, \citenamefont {Parisi},\ and\ \citenamefont
  {Zhang}(1986)}]{kardar1986dynamic}%
  \BibitemOpen
  \bibfield  {author} {\bibinfo {author} {\bibfnamefont {M.}~\bibnamefont
  {Kardar}}, \bibinfo {author} {\bibfnamefont {G.}~\bibnamefont {Parisi}}, \
  and\ \bibinfo {author} {\bibfnamefont {Y.-C.}\ \bibnamefont {Zhang}},\
  }\bibfield  {title} {\enquote {\bibinfo {title} {Dynamic scaling of growing
  interfaces},}\ }\href@noop {} {\bibfield  {journal} {\bibinfo  {journal}
  {Physical Review Letters}\ }\textbf {\bibinfo {volume} {56}},\ \bibinfo
  {pages} {889} (\bibinfo {year} {1986})}\BibitemShut {NoStop}%
\bibitem [{\citenamefont {Nakao}\ and\ \citenamefont
  {Mezi{\'c}}(2018)}]{nakao2018sice}%
  \BibitemOpen
  \bibfield  {author} {\bibinfo {author} {\bibfnamefont {H.}~\bibnamefont
  {Nakao}}\ and\ \bibinfo {author} {\bibfnamefont {I.}~\bibnamefont
  {Mezi{\'c}}},\ }\bibfield  {title} {\enquote {\bibinfo {title} {Koopman
  eigenfunctionals and phase-amplitude reduction of rhythmic reaction-diffusion
  systems},}\ }in\ \href@noop {} {\emph {\bibinfo {booktitle} {Proceedings of
  the SICE Annual Conference 2018 September 11-14, 2018, Nara, Japan}}}\
  (\bibinfo {year} {2018})\ pp.\ \bibinfo {pages} {74--77}\BibitemShut
  {NoStop}%
\bibitem [{\citenamefont {Nakao}(shed)}]{nakao2020springer}%
  \BibitemOpen
  \bibfield  {author} {\bibinfo {author} {\bibfnamefont {H.}~\bibnamefont
  {Nakao}},\ }\bibfield  {title} {\enquote {\bibinfo {title} {Phase and
  amplitude description of complex oscillatory patterns in reaction-diffusion
  systems},}\ }in\ \href@noop {} {\emph {\bibinfo {booktitle} {The Physics of
  Biological Oscillators}}}\ (\bibinfo  {publisher} {Springer},\ \bibinfo
  {year} {2020, to be published})\BibitemShut {NoStop}%
\bibitem [{\citenamefont {Lu}(1991)}]{lu1991hartman}%
  \BibitemOpen
  \bibfield  {author} {\bibinfo {author} {\bibfnamefont {K.}~\bibnamefont
  {Lu}},\ }\bibfield  {title} {\enquote {\bibinfo {title} {A
  {H}artman-{G}robman theorem for scalar reaction-diffusion equations},}\
  }\href@noop {} {\bibfield  {journal} {\bibinfo  {journal} {Journal of
  differential equations}\ }\textbf {\bibinfo {volume} {93}},\ \bibinfo {pages}
  {364--394} (\bibinfo {year} {1991})}\BibitemShut {NoStop}%
\bibitem [{\citenamefont {Dieudonn{\'e}}(2011)}]{dieudonne2011foundations}%
  \BibitemOpen
  \bibfield  {author} {\bibinfo {author} {\bibfnamefont {J.}~\bibnamefont
  {Dieudonn{\'e}}},\ }\href@noop {} {\emph {\bibinfo {title} {Foundations of
  modern analysis}}}\ (\bibinfo  {publisher} {Read Books Ltd},\ \bibinfo {year}
  {2011})\BibitemShut {NoStop}%
\bibitem [{\citenamefont {Zeidler}(2012)}]{zeidler2012applied}%
  \BibitemOpen
  \bibfield  {author} {\bibinfo {author} {\bibfnamefont {E.}~\bibnamefont
  {Zeidler}},\ }\href@noop {} {\emph {\bibinfo {title} {Applied functional
  analysis: applications to mathematical physics}}},\ Vol.\ \bibinfo {volume}
  {108}\ (\bibinfo  {publisher} {Springer Science \& Business Media},\ \bibinfo
  {year} {2012})\BibitemShut {NoStop}%
\end{thebibliography}
\end{document}